%% file: main.tex
\documentclass[sigconf,nonacm]{acmart}

\usepackage{multirow}
\usepackage{tabularx}
\usepackage{amsmath}
\usepackage{tikz}
\usepackage{pgfplots}
\pgfplotsset{compat=1.16}
\usepackage[ruled,vlined,linesnumbered]{algorithm2e}
\usepackage{algpseudocode}
\usepackage{subcaption}
\usepackage{xcolor} 

\usepackage{listings}
\usepackage{xcolor}

\definecolor{etalePurple}{HTML}{4B2E83}
\definecolor{etaleGold}{HTML}{85754D}
\definecolor{etaleComment}{HTML}{6B6B6B}
\definecolor{etaleBg}{HTML}{FBFAFD}

\lstdefinestyle{etalecpp}{
  language=C++,
  basicstyle=\ttfamily\footnotesize,
  keywordstyle=\color{etalePurple}\bfseries,
  commentstyle=\color{etaleComment}\itshape,
  stringstyle=\color{etaleGold},
  numberstyle=\tiny\color{etaleComment},
  backgroundcolor=\color{etaleBg},
  frame=single,
  rulecolor=\color{etalePurple},
  framesep=6pt,
  breaklines=true,
  showstringspaces=false,
  columns=fullflexible,
  keepspaces=true,
  morekeywords={uint32_t,uint64_t,int,float,bool,constexpr,struct,using,size_t,volatile,inline,return,const}
}

\newtheorem{proposition}{Proposition}

\theoremstyle{definition}

\theoremstyle{remark}
\newtheorem{remark}{Remark}

\usepackage{pifont}
\newcommand{\cmark}{\ding{51}}
\newcommand{\xmark}{\ding{55}}

\begin{document}

\title{SHEAF: Self-profiled Hardness Estimation from Answer-set Flux for Predicting Query Hardness in Graph-based ANN Search}

\author{Dongfang Zhao (dzhao@uw.edu)}
\affiliation{%
  \city{}
  \country{}
}

\keywords{Approximate nearest neighbor search, graph-based index, query hardness
  estimation, local intrinsic dimensionality}

\begin{abstract}
Graph-based approximate nearest neighbor (ANN) search is usually governed by a beam-width parameter
that trades recall for throughput and is fixed for the whole workload. 
Yet, queries may not be equally hard: for example, on the widely used data set SIFT1M,
the beam that a query needs to reach 95\% recall varies by more than $32\times$. 
Therefore, serving each query at its own width would help if
the system could tell, cheaply and in advance, how hard it is. 
The prevailing proxy for this difficulty is called local intrinsic dimensionality (LID);
however, LID is static and geometric, which makes it only weakly predict the minimum beam.

This paper presents a new measure, namely Self-profiled Hardness Estimation from Answer-set Flux (SHEAF),
which represents a query's hardness as how much its own top-$k$ answer set changes between two shallow probe widths.
We design a self-profiling estimator that turns this flux into a deployable per-query beam predictor;
furthermore, we develop a fixed-probe evaluation protocol that scores each measure over all queries with an observed minimum sufficient beam. 
On popular ANN indexes such as CAGRA and HNSW across four diverse data sets,
SHEAF predicts the per-query beam better than five baseline measures on both GPU and CPU by up to $1.55\times$ in held-out correlation, using only two shallow probe searches and no query-time ground truth. 
\end{abstract}

\maketitle

\section{Introduction}
\label{sec:intro}
 
Graph-based approximate nearest neighbor (ANN) search is usually governed by multiple parameters;
one of the most important parameters is the search width, also known as the beam width.
On GPUs, cuVS CAGRA's \texttt{itopk\_size}~\cite{cagra}, like HNSW's \texttt{ef}~\cite{hnsw} on the
CPU, sets the beam width that trades recall for throughput.
In practice, the beam width is tuned once to place the
whole workload at a single point on the recall-throughput curve. 
However, this one-size-fits-all
setting ignores the fact that many queries in the real-world are not equally hard. 
As shown in Figure~\ref{fig:intro}(a), our experiments show that on CAGRA over SIFT1M~\cite{sift}, the minimum
\texttt{itopk} a query needs to reach recall@10 $\ge 0.95$ spans more than $32\times$ with
$61\%$ of queries satisfied at the smallest beam and a long tail needing an order of magnitude
more. 
Therefore, serving each query at its own width, i.e., small for the easy
majority and large only for the hard tail, would move the system off that uniform curve. Doing so
rests on one ingredient above all: a way to tell, cheaply and before committing the search
budget, how hard each query is. 

\begin{figure}[t]
  \centering
  \includegraphics[width=\columnwidth]{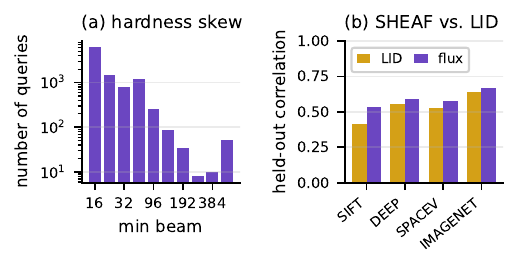}
  \caption{(a) Minimum beam on CAGRA (SIFT1M) spans $32\times$. (b) Flux outperforms LID at predicting
  query hardness.}
  \label{fig:intro}
\end{figure} 

One prevailing way to price a query is the so-called local intrinsic dimensionality
(LID)~\cite{houle_lid, amsaleg_lid}, 
which is a static geometric measure for
grading a query's difficulty~\cite{aumuller_lid_bench}. 
However, recent studies showed that LID-based approaches exhibit two limitations.
First, LID is a property of the data alone that is computed independently of the graph index, whereas on a
graph the effort to answer a query is governed by connectivity.
To fix it, the recent Steiner-hardness measure was proposed~\cite{steiner_hardness};
however, Steiner-hardness is an offline, ground-truth-dependent graph-native reference and requires each query's ground-truth neighbors,
which can be only obtained after the search rather than during it.
Second, existing adaptive-search methods read static data features~\cite{tao} or the running distance
sequence~\cite{li_adaptive_et, adaptive_beam}, but they do not directly correlate the stability of the top-$k$ answer set itself, which is a direct indicator of how much more work a query needs.
What is desired is a cheap, online,
ground-truth-free measure that actually quantifies a query's graph-search cost, i.e., beam width.

Our insight toward this goal is that a query's hardness is not only geometric but algorithmic: a query is
hard exactly when its own search has not yet settled, and more importantly, that state is already visible in the
search's early behavior without ground truth at query time. We formulate it as \emph{answer-set flux}, how much a
query's top-$k$ result changes between two shallow probe widths. A query whose
answer set is still churning is far from settled; one whose answer set has stopped moving, and
whose distances have stopped tightening, is a candidate for convergence that the estimator then
confirms.
Performance-wise, this flux is computed solely from the results of two shallow probe searches
without ground truth or access to the graph index.

This paper turns the above insight into a new quantitative measure, namely Self-profiled Hardness Estimation from Answer-set Flux (SHEAF). 
What
SHEAF contributes is a query-hardness measure and the machinery to read and evaluate it: flux
itself, a self-profiling estimator that turns it into a per-query cost prediction, and a protocol
that scores that prediction with a probe held below the hard tail's cost. An adaptive serving
policy is the application such a measure feeds; how accurately the measure predicts per-query
cost, which is what we study here, is what any such policy would rest on. 
In summary, this paper makes the following contributions.

\paragraph{We define answer-set flux as a dynamic measure of query hardness (\S\ref{sec:method:flux})}
SHEAF runs two searches for the same query. The first search uses a
  beam $w_0$. The second search uses a larger beam $w_1$, where $w_0<w_1$. Every query uses the
  same $w_0$ and $w_1$. Both widths lie at the low end of the beam ladder to keep the two searches
  cheap. We call these two searches the \emph{probe}. Flux is the fraction of the top-$k$ results at
  $w_0$ that are replaced at $w_1$. High
  flux means that the search result changes when more search budget is used. Low flux means that
  the result remains stable. Flux needs neither ground truth nor access to the index internals.
  A short set-counting argument shows that this simple yet effective flux upper-bounds the recall change between the two
  beam widths (Prop.~\ref{prop:churn}).

  \paragraph{We design a self-profiling estimator of per-query search cost (\S\ref{sec:method:estimator})}
  We train a lightweight regressor that takes flux and the
  returned distances from the probe as input. The regressor predicts the minimum beam needed for
  a target recall. 
  At query time, SHEAF runs the probe and predicts the required beam before the
  full search begins. Since every query uses the same $w_0$ and $w_1$, the relative
  probe cost approaches zero as the required beam increases, an asymptotic property we state as
  Prop.~\ref{prop:overhead}. These inputs come from standard top-$k$
  search results and do not require graph internals. The same estimator design therefore applies
  across hardware platforms, e.g., from GPU-based CAGRA to CPU-based HNSW.

\paragraph{We develop a fixed-probe protocol for evaluating search-based hardness (\S\ref{sec:method:eval})}
We use the same $w_0$ and $w_1$ for every query and evaluate each
  measure over all answerable queries. For the
  hard queries whose required beam is larger than $w_1$, the probe is guaranteed to not have reached the target recall. We formalize this property as Prop.~\ref{prop:decirc}.

We evaluate SHEAF across four diverse data sets against five baseline hardness measures. 
SHEAF outperforms every baseline on both GPU and CPU, with a margin over the prevailing geometric baseline, exact LID, reaching up to $1.55\times$ in held-out correlation.
As an example, Figure~\ref{fig:intro}(b) shows the advantage of SHEAF over the conventional LID.
 
\section{Related Work}
\label{sec:related}

\paragraph{Approximate nearest neighbor (ANN) search}
ANN research trades exactness for speed, and most methods fall into one of three families.
First, quantization compresses vectors so that distances are approximated cheaply, ranging from
product quantization~\cite{sift} and its optimized~\cite{tge_tpami13} and additive~\cite{jmart_eccv18}
variants to anisotropic~\cite{rguo_icml20} and recent theory-grounded~\cite{jgao_sigmod24} schemes,
consolidated in libraries such as Faiss~\cite{mdouz_arxiv24}. 
Second, locality-sensitive hashing
gives sublinear query time with probabilistic guarantees~\cite{pindy_stoc98, mdata_socg04}, refined
by multi-probe~\cite{qlv_vldb07} and near-optimal constructions~\cite{aando_neurips15}. 
Third,
space-partitioning trees~\cite{jbent_cacm75} and their auto-tuned toolkits~\cite{mmuja_tpami14}
remain common baselines. These indexes now underpin many applications, for example vector
databases~\cite{jwang_sigmod21, cwei_vldb20} and retrieval-augmented and recommendation
pipelines~\cite{plewi_neurips20, rying_kdd18}.

\paragraph{Graph-based indexes and the width knob}
Graph indexes have become one of dominant ANN paradigms.
They could build on navigable small-world graphs~\cite{ymalk_is14, hnsw} with monotonic or
spreading-out designs~\cite{cfu_vldb19, bharw_cvpr16}. 
Methods such as DiskANN~\cite{diskann},
HCNNG~\cite{jmuno_pr19}, and EFANNA~\cite{cfu_arxiv16} answer a query by a best-first walk over a
$k$-NN graph built by NN-descent and its variants~\cite{wdong_www11, nono_mm23}. Recent work adds
parallel and deterministic construction at scale~\cite{mmano_ppopp24}, sharper graph
convergence~\cite{li_convergent26}, and a theoretical analysis of greedy graph
search~\cite{prokhorenkova_theory20}. 
Across all of them, however, a single beam-width knob (CAGRA's \texttt{itopk\_size}, HNSW's \texttt{ef})~\cite{cagra}
trades recall for speed and is normally fixed for the whole workload, which places the system at one
point on the recall-throughput curve. SHEAF instead serves each query at its own width; this is the
premise that makes a per-query notion of hardness necessary.

\paragraph{GPU, large-scale, and dynamic indexes}
A parallel line of work scales graph ANN to modern hardware and workloads. 
On GPUs, SONG~\cite{wzhao_icde20}, CAGRA~\cite{cagra}, and GTS~\cite{yzhu_sigmod24}
exploit batched parallelism, joined by recent GPU construction and billion-scale search
systems~\cite{hwang_cikm21, zli_sigmod25, kvenk_tbd25, dzhao_hpdc26}. Disk-resident and streaming
variants~\cite{chen_spann21, mwang_sigmod24, ypan_bigdata23, yxu_sosp23, asing_arxiv21},
with dynamic or
attribute-filtered indexes~\cite{digra, qzhan_osdi23, lpate_sigmod24, sgoll_www23} keep
billion-scale and evolving datasets searchable. These systems, however, optimize how a chosen width is served, not how the width is chosen per query; the latter is the gap SHEAF addresses.
 
\paragraph{Intrinsic dimensionality and query hardness}
The prevailing way to reason about query difficulty is geometric, namely through local intrinsic
dimensionality (LID)~\cite{houle_lid, amsaleg_lid}. LID is the field's default hardness axis, and it is
used to stratify ANN benchmarks into easy and hard query sets~\cite{aumuller_lid_bench}. Its standard
structural counterpart is hubness, i.e., the tendency of a few points to dominate neighbor lists in
high dimension~\cite{radovanovic_hubness}. 
Closest to our diagnosis is
Steiner-hardness~\cite{steiner_hardness}, which argues that a query's cost is governed by
connectivity rather than density.
Steiner-hardness measures the query's cost as an NP-hard Directed-Steiner-Tree quantity
over a query's ground-truth neighbors; complementary theory characterizes the hardness of
navigable graphs themselves~\cite{khanna_navigable25}. SHEAF departs from all of these. We test
LID and hubness as predictors of actual per-query cost on real CAGRA, where both prove suboptimal.
We will use Steiner-hardness only as an offline, ground-truth-dependent graph-native
reference, since it cannot price a query at search time.

\paragraph{Adaptive and per-query search}
A final line of work makes the search itself adaptive, which is where SHEAF sits; however, it reads
different quantities. Learned routing augments graph vertices to escape local
minima~\cite{baranchuk_route19} and query-aware modules pick a better entry point per
query~\cite{jruan_kdd25}. 
Other work includes: Learned early-termination predicts when to stop from features of the
running search~\cite{li_adaptive_et}, Tao selects parameters from static data
features~\cite{tao}, and distance-adaptive beam search replaces the fixed stopping rule with a
relative-distance criterion~\cite{adaptive_beam}. 
Recent systems even detect and repair hard or
out-of-distribution queries on the fly~\cite{hua_hardness26}. 
SHEAF differs from all of these in
what it measures, not in the control policy. 
These methods learn from static features or read the
running distance sequence; SHEAF instead reads the stability of the top-$k$ answer
itself, i.e., how much it still changes as the probe widens. This is a direct, ground-truth-free
measure of how much the query's search has settled. 
To our knowledge, prior work has not explicitly applied a
fixed-probe control when evaluating such self-observed features, including a pre-target guarantee
on the hard tail.

\section{Methodology}
\label{sec:method}

The quantity that a measure predicts is a query's \emph{cost}: the smallest beam
width at which its own search reaches the recall target, e.g., recall@10 = 0.95. 
Formally, fix a base set $X\subset\mathbb{R}^d$,
a graph index built on it, and a neighborhood size $k$; for a query $q$, let $N=N_k(q)$ be its
true $k$ nearest neighbors. The graph search at beam width $w$ returns a top-$k$ set $R_w(q)$ with $k$-th returned distance
$r_w(q)$ and recall $\rho_w(q)=|R_w(q)\cap N|/k$. Fix a recall target $\tau$ and a finite width
  ladder $L=\{b_1<b_2<\cdots<b_m\}$ and let
  $A(q)=\{b\in L:\rho_b(q)\ge\tau\}$. For a query with $A(q)\ne\varnothing$, its cost is
\begin{equation}
  c(q)=\min A(q),
  \label{eq:cost}
\end{equation}
the minimum beam that reaches the target. We evaluate cost prediction only on these answerable
queries. Queries with $A(q)=\varnothing$ are right-censored (i.e., its true cost lies beyond the largest width) because no sufficient beam is observed
within the ladder. We exclude them from the correlation.


\subsection{Dynamic Difficulty Measure}
\label{sec:method:flux}

We read difficulty from the search's own early behavior rather than from the query's geometry.
The \emph{probe} is a cheap shallow search: the same graph search runs at two small widths
$w_0<w_1$. Both widths lie far below what hard queries need. The probe returns the top-$k$ sets
$R_0=R_{w_0}(q)$ and $R_1=R_{w_1}(q)$ together with their $k$-th distances $r_0$ and $r_1$.

\emph{Answer-set flux} summarizes how much the probe is still moving. Its primary component is
the \emph{set churn}
\begin{equation}
  \chi(q)=\frac{|R_0\triangle R_1|}{2k}\in[0,1],
  \label{eq:churn}
\end{equation}
where $R_0\triangle R_1$ is the symmetric difference of the two answer sets, namely the members
in exactly one of them. Churn is thus the fraction of the top-$k$ that changes as the beam widens
from $w_0$ to $w_1$. Flux also reads the \emph{distance improvement} $(r_0-r_1)/r_0$ and the probe's raw $k$-th and mean
distances at $w_0$. A query whose answer has stopped moving and whose distances have stopped
tightening looks close to convergence, and the estimator of \S\ref{sec:method:estimator} treats
it as such only when both cues agree. A churning
answer instead marks a query far from settled. All four flux features are byproducts of a search
the index already runs. None needs ground truth at query time.
 
Beyond correlating with convergence, churn upper-bounds the recall that the search can
still gain over the probe interval. 
Although the bound is seemingly elementary, it turns a heuristic into a
guarantee: a stable answer caps the recall that the interval can still contribute, whereas a moving
answer leaves that headroom open. The bound assumes nothing about the data distribution
or the graph. Moreover, it holds for every query and every pair of probe widths. With such generality, a two-width probe serves as a proxy for a property of the entire search. The following proposition states this formally.

\begin{proposition}[Churn bounds recall change over probe interval]
  \label{prop:churn}
  For any query and widths $w_0<w_1$ with $|R_0|=|R_1|=k$,
  \[
    \bigl|\rho_{w_1}(q)-\rho_{w_0}(q)\bigr|\ \le\ \frac{|R_0\triangle R_1|}{2k}\ =\ \chi(q).
  \]
\end{proposition}

\begin{proof}[Proof sketch]
  Because $|R_0|=|R_1|=k$, the sets differ symmetrically:
  $|R_1\setminus R_0|=|R_0\setminus R_1|=k-|R_0\cap R_1|=|R_0\triangle R_1|/2$. The common
  part $R_0\cap R_1$ contributes equally to both recalls and cancels, so
  $\bigl||R_1\cap N|-|R_0\cap N|\bigr|\le|R_1\setminus R_0|$. Dividing by $k$ gives the
  claim. The complete proof can be found in Appendix~\ref{app:churn}.
\end{proof}

We run a small example to make the bound concrete. For $k=10$, a probe that changes two of the ten
answers has $|R_0\triangle R_1|=4$ and $\chi=0.2$.
Proposition~\ref{prop:churn} caps the recall still to be gained across the probe interval at
$0.2$, whether or not those two changes are true neighbors. A query that changes nothing
($\chi=0$) has provably gained no recall over the interval $[w_0,w_1]$ itself. This is a local
certificate, not a global one: it rules out recall gained inside the probe interval, but it does
not by itself say the search has reached the target $\tau$, since a stable answer can also mean
the probe has not yet begun to discriminate at all rather than that it has finished.

Nonzero churn is the weaker direction of this bound. A moving answer means recall \emph{may}
still change, yet the bound does not guarantee that it will, since the churn could reshuffle
non-neighbors and leave recall flat. \S\ref{sec:method:estimator} will show how the estimator
combines flux with the probe's distance features. A plateau is then trusted only when the fuller
prediction agrees, which keeps an early, uninformative probe from being mistaken for a converged
one.

Among these features, churn is the most important, while the distance terms are secondary. Churn
measures a change that the distances cannot see: two queries can have identical distance profiles
yet differ in how much their top-$k$ \emph{membership} shifts between the two probes. This shift
tracks a query's unfinished work. The distance terms still help on datasets whose operating beam
lies far beyond the probe, where they add what churn does not capture.

By default the probe uses $w_0{=}16$ and $w_1{=}24$. These widths are large enough for an answer
set to begin forming on most datasets, but small enough to stay below what a hard query needs. On
the hard tail, the probe therefore reports genuine non-convergence rather than a query's final
answer. A wider second width $w_1{=}32$ raises the correlation slightly. It also overlaps the low
end of the ladder $L$ and shrinks the hard tail on which the probe stays provably below the target.
The default therefore keeps $w_1{=}24$.

A dynamic measure can succeed where geometry fails. The beam that a graph search needs is a property
of the graph, not of the embedding alone. A geometric proxy such as LID scores only the query's
neighborhood in the embedding space. The required beam instead depends on how that neighborhood
is connected, where the search enters, and how fast greedy expansion reaches the true neighbors.
Two queries of equal local dimension can therefore sit in very differently navigable parts of
the same index. Rather than modeling any of this directly, flux reads how much the search's own
answer is still moving, which captures the graph's geometry and the search dynamics at once,
without naming either.
This is why a quantity from two shallow searches can carry information that an \emph{exact}
geometric measure does not.

\subsection{Self-Profiling Difficulty Estimator}
\label{sec:method:estimator}

The estimator turns flux from a measurement into a per-query beam prediction before the search
commits (Algorithm~\ref{alg:estimate}). The probe features make up a vector
\begin{equation}
  \phi(q)=\bigl(\chi(q),\ (r_0-r_1)/r_0,\ r_0,\ \bar r_0\bigr),
  \label{eq:feat}
\end{equation}
with $\bar r_0$ the probe's mean returned distance. A light predictor $h$ of log-cost,
$\hat c(q)=\exp h(\phi(q))$, is fit on a held-out split of queries. The predictor is calibrated offline
on labeled training queries. Only query-time feature extraction is ground-truth-free. 
At query time, the estimator probes
at $w_0$ and $w_1$ and always forms the prediction $\hat c(q)$ from the full feature vector, not
from churn alone. If $\hat c(q)$ is already at or below $w_1$, the probe result is returned
directly; otherwise the query is searched at the smallest ladder width at least $\hat c(q)$.
This shortcut is gated on the prediction rather than on churn by itself, because a
zero-churn probe can mean the query has converged, or it can mean the probe has not yet begun to
discriminate (\S\ref{sec:eval:ablation}), and only the full prediction, built from churn and the
probe's distance features together, can tell the two apart. The estimator touches the index
only through the probe. That probe is just two top-$k$ queries at fixed widths. These two facts
together make it portable and cheap.

\begin{algorithm}[t]
\caption{Self-profiling adaptive search}\label{alg:estimate}
\KwIn{query $q$; index; widths $w_0<w_1$; ladder $L$; predictor $h$}
\KwOut{approximate $k$ nearest neighbors of $q$}
\ForEach{$w\in\{w_0,w_1\}$}{
  $R_w\leftarrow\textsc{Search}(q,w)$\;
  $r_w\leftarrow\textsc{KthDist}(R_w)$\;
}
$\chi\leftarrow|R_{w_0}\triangle R_{w_1}|/(2k)$\;
$\phi\leftarrow\bigl(\chi,\ (r_{w_0}-r_{w_1})/r_{w_0},\ r_{w_0},\ \textsc{MeanDist}(R_{w_0})\bigr)$\;
$\hat c\leftarrow\exp\bigl(h(\phi)\bigr)$\;
\If{$\hat c\le w_1$}{
  \Return $R_{w_1}$\;
}
$b\leftarrow\max L$\;
\ForEach{$b'\in L$ in increasing order}{
  \If{$b'\ge\hat c$}{
    $b\leftarrow b'$\;
    \textbf{break}\;
  }
}
\Return \textsc{Search}(q,b)\;
\end{algorithm}

Algorithm~\ref{alg:estimate} adds only a fixed, query-independent cost on top of the search that a
query runs anyway. Let $W(w)$ denote the per-query search work at beam width $w$; this analysis
assumes a linear cost model, $W(w)=\Theta(w)$, which both CAGRA's \texttt{itopk} search and
HNSW's \texttt{ef} search exhibit by the same bounded-candidate-list argument. Under
this model, the two probes cost $W(w_0)+W(w_1)=\Theta(w_1)$. The
churn, feature vector, and predictor $h$ are $O(k)$ over the returned $k$-sets and $O(1)$ in the
fixed-length feature; forming $\hat c(q)$ costs the same $O(1)$ whether or not it ends up at or
below $w_1$. Choosing the ladder rung is $O(|L|)$. When $\hat c(q)>w_1$, the
committed search dominates at $W(b)=\Theta(b)$ for the predicted beam $b\approx c(q)$. The
per-query time is then $\Theta(w_1+b)$; the probe is an additive $\Theta(w_1)$ overhead. The
per-query time is $\Theta(w_1)$ when $\hat c(q)\le w_1$ and the probe result is returned directly. Space is $O(k)$
for the two probe result sets.

This fixed probe cost has two consequences that make the estimator practical beyond a single
system. First, the probe reads only the width-parameterized top-$k$ that every graph index
already exposes. The same estimator design therefore applies to any of them, refit but otherwise
unchanged. Second, since its cost
stays $\Theta(w_1)$ while a hard query pays $\Theta(c(q))$ for the search that it needs anyway, the
probe is asymptotically free exactly where queries are expensive. The following proposition
formalizes these two properties.

\begin{proposition}[Fixed-cost, index-agnostic estimation]
  \label{prop:overhead}
  Let $W(w)$ be the per-query search work at width $w$, nondecreasing in $w$.
  (i) The estimator is well defined for any ANN index exposing a width-parameterized top-$k$,
  regardless of $W$.
  (ii) If in addition $W(w)=\Theta(w)$, a linear-cost condition that both CAGRA's itopk search
  and HNSW's ef search satisfy, then the probe cost is $W(w_1)=\Theta(w_1)$ per query, independent of the query's
  cost label $c(q)$. The relative overhead $W(w_1)/W(c(q))=\Theta\!\bigl(w_1/c(q)\bigr)$ therefore tends to $0$
  on the hard tail $c(q)\gg w_1$.
\end{proposition}

\begin{proof}[Proof sketch]
  (i) $\phi$ is a function only of $R_{w_0}$, $R_{w_1}$, and their returned distances, all
  emitted by any width-parameterized search; no index internals appear and no cost model is
  assumed, which is why the
  same feature map and fitting procedure transfer from the GPU CAGRA graph to the CPU HNSW graph
  unchanged, each refit on that index's own queries
  (\S\ref{sec:eval:generalize}). 
  (ii) The probe issues two searches at fixed $w_0<w_1$;
  under $W(w)=\Theta(w)$, its cost is $\Theta(w_1)$ regardless of $q$, while a query that
  needs width $c(q)$ pays $\Theta(c(q))$.
  Appendix~\ref{app:overhead} gives the counting and notes an unimplemented warm-start
  refinement in which a resumable search could reuse part of the probe work.
\end{proof}

The predictor is fit against $\log c(q)$ rather than $c(q)$. Per-query cost may span more than an order of
magnitude across a workload (recall Figure~\ref{fig:intro}(a)). A linear fit on raw cost would therefore be
dominated by the hard tail. The result depends on the features $\phi$ far more than on the
regressor that reads them. In our implementation, $h$ is an ordinary least-squares fit on
$\phi$ with no regularization or hyperparameter to tune;
\S\ref{sec:eval:flux} will show that the
flux-over-LID ordering is preserved when it is swapped for a $k$-NN
regressor.

On a batched GPU index, the predicted widths are consumed by grouping. Queries whose $\hat c(q)$
falls in the same ladder band are dispatched together at that band's \texttt{itopk}. The easy
majority then runs at a small width, and only the predicted-hard tail pays for a large beam. The
estimator therefore needs no per-query kernel launch; it only reorganizes an existing batched
search. A prediction need only be accurate enough to place a query in the right band, and by
Proposition~\ref{prop:overhead} the probe that produces it becomes negligible relative to the
committed search on the sufficiently hard bands, though not on bands near the probe width itself.

A misprediction does not invalidate the search procedure, but it can cost either recall or
efficiency. The search still runs to the ladder width that the estimator selects, and a query's
recall is whatever that width delivers, never silently corrupted below it. A too-narrow beam
places a query below the rung that it needed. By the minimality in the cost definition~\eqref{eq:cost},
any ladder rung below $c(q)$ misses the target recall: this is the recall cost
of an error. A too-wide beam instead selects a
rung above the minimum sufficient width and spends more work than needed. This extra work
cannot exceed what a bounded number of adjacent ladder rungs cost, which is the efficiency
cost of an error. A proposition at the close of this section will bound this extra work
(Proposition~\ref{prop:misprediction}), once the fixed-probe protocol below has introduced
the notation it shares. 
 
 
\subsection{Fixed-Probe Evaluation Protocol}
\label{sec:method:eval}

A self-observed measure is computed from the search's own behavior, which raises a risk that a
static proxy does not share: two such measures could each look accurate simply because they
agree with each other, while both miss a query's real cost. The protocol
(Algorithm~\ref{alg:protocol}) guards against this by scoring every measure against the
query's true cost label on a held-out split of queries that the measure was not fit on.

Two further rules keep the comparison fair between measures. Every baseline is given its
strongest available form, for example, exact LID computed from the true $k$-NN rather than an
online approximation. Any margin that flux shows is therefore not an artifact of a weakened
baseline. Every online measure, flux included, is read from the same fixed, narrow probe at the
low end of the beam ladder. No measure therefore sees a wider or more informative window into
the search than another. This fixed probe keeps the comparison uniform across all answerable
queries and stays below the target recall on the hard tail.

\begin{algorithm}[t]
\caption{Fixed-probe evaluation of a measure $m$}\label{alg:protocol}
\KwIn{queries $Q$; index; ladder $L$; target $\tau$; fixed probe widths $w_0<w_1$; measure $m$}
\KwOut{held-out correlation between predicted and true $\log$-cost}
\ForEach{$q\in Q$}{
  \ForEach{$b\in L$ in increasing order}{
    $R_b\leftarrow\textsc{Search}(q,b)$\;
    $\rho_b\leftarrow|R_b\cap N_k(q)|/k$\;
  }
  $A_q\leftarrow\{b\in L:\rho_b\ge\tau\}$\;
  $x_q\leftarrow m(q)$ from the fixed probe $(w_0,w_1)$\;
}
$Q_{\mathrm{obs}}\leftarrow\{q\in Q:A_q\ne\varnothing\}$\;
$c(q)\leftarrow\min A_q$ for each $q\in Q_{\mathrm{obs}}$\;
partition $Q_{\mathrm{obs}}$ into $Q_{\mathrm{fit}}\sqcup Q_{\mathrm{test}}$\;
fit predictor $h_m$ of $\log c(q)$ on $\{(x_q,\log c(q)):q\in Q_{\mathrm{fit}}\}$\;
\Return correlation of $h_m(x_q)$ with $\log c(q)$ over $q\in Q_{\mathrm{test}}$\;
\end{algorithm}

Algorithm~\ref{alg:protocol} is a procedure run once over the query set and dominated
by the cost of recovering the ground-truth label. For each of the $|Q|$ queries, it searches the
full width ladder. That costs $\sum_{b\in L}W(b)=\Theta(\sum_{b\in L}b)$ under the linear model,
plus $O(|L|\,k)$ to score recall against the true neighbors $N_k(q)$. Recovering $c(q)$ and
reading the measure $m$ add only $O(|L|)$ and $O(1)$. The split, the regression on a
fixed-length feature, and the held-out correlation are each linear in $|Q|$. The total is thus
$O\!\bigl(|Q|\,(\sum_{b\in L}W(b)+|L|\,k)\bigr)$ time and $O(|Q|)$ space. This one-time offline
cost recovers the per-query labels $c(q)$; because the probe stays fixed and narrow,
Proposition~\ref{prop:decirc} will demonstrate that it is pre-target on the hard tail. The online
estimator of Algorithm~\ref{alg:estimate} pays none of it.

The cost label is recovered from ordinary batched searches.
For example, CAGRA applies one \texttt{itopk} to
a batch but returns per-query neighbors. The protocol therefore reads $c(q)$ in \eqref{eq:cost} from a
handful of batched searches at increasing ladder widths, taking the smallest width whose
per-query result attains $\rho\ge\tau$. This needs no kernel instrumentation and no per-query
reruns. Each measure is then scored by held-out prediction of $\log c(q)$, fit on a random half
of the queries and reported on the other half. Agreement among proxies is not evidence: two
proxies can correlate strongly with each other yet jointly miss cost, as any feature $u$ and its
affine copy $\alpha u+\beta$ do. Only out-of-sample prediction of the cost itself counts.

Each baseline is given its best case. Any weakness is therefore the measure's own rather than an
estimator's. Local intrinsic dimensionality is tested in its \emph{exact} form, computed from
true $k$-NN distances via the Hill/MLE estimator~\cite{houle_lid,amsaleg_lid} rather than any
online approximation. Hubness, the structural alternative, is the in-degree of a query's true
neighbors in the graph~\cite{radovanovic_hubness}. Steiner-hardness needs each query's
ground-truth neighbors and an NP-hard solve. Steiner-hardness enters only as an offline, ground-truth-dependent
graph-native reference~\cite{steiner_hardness}.

The fixed-probe control sets $w_0<w_1$ at the low end of the ladder for every
query. The evaluation scores each measure over all answerable queries. The fixed narrow probe provides a guarantee
on the queries that matter: whenever a query's cost label exceeds $w_1$, the probe has not reached
the target. The measure therefore predicts beyond the observed probe widths. Since
$w_1$ sits far below the beam that the hard tail needs, this covers exactly the hard, expensive
queries where prediction is both difficult and valuable. The following proposition makes the
guarantee precise.

\begin{proposition}[Probe is pre-target on hard tail]
  \label{prop:decirc}
  Suppose $q$ is answerable and $w_1\in L$ with $w_1<c(q)$. Then the probe has not reached the target recall:
  $\rho_{w_1}(q)<\tau$.
\end{proposition}

\begin{proof}[Proof sketch]
  This follows from \eqref{eq:cost}. Since $q$ is answerable, $w_1<c(q)$ places $w_1$ below
  the minimum width that reaches $\tau$. Therefore $\rho_{w_1}(q)<\tau$.
  The complete proof can be found in Appendix~\ref{app:decirc}.
\end{proof}

Algorithm~\ref{alg:protocol} does not name any specific measure; it only takes a quantity
$m(q)$ and scores it. Exact LID, hubness, Steiner-hardness, and flux are each such a quantity,
computed on their own and then handed to the same protocol.
The pre-target guarantee of Proposition~\ref{prop:decirc} is therefore not unique to flux
either. Its proof depends only on the condition $w_1<c(q)$, not on what $m(q)$ computes. It
holds for every measure scored under this protocol.

We close this section with the bound promised in \S\ref{sec:method:estimator}: what it
costs when the estimator serves a wider beam than a query needs. The bound uses
the same ladder $L$ and cost $c(q)$ defined above. 

\begin{proposition}[Extra work from too-wide beam is bounded by ladder granularity]
  \label{prop:misprediction}
  Let $L=\{b_1<\cdots<b_m\}$ and fix a query $q$ with $c(q)=b_k$. If the estimator serves rung
  $b_{k+j}$ for some $j\ge 0$, the excess work relative to serving at $c(q)$ satisfies
  \[
    W(b_{k+j})-W(c(q))\ \le\ j\cdot\max_{1\le i<m}\bigl(W(b_{i+1})-W(b_i)\bigr).
  \]
\end{proposition}

\begin{proof}[Proof sketch]
  Telescope $W(b_{k+j})-W(b_k)$ into $j$ consecutive rung-to-rung differences and bound each by
  the largest such gap on $L$.
  The complete proof can be found in Appendix~\ref{app:misprediction}.
\end{proof}

\section{Evaluation}
\label{sec:eval}


\subsection{Experimental Setup}
\label{sec:eval:setup}

\paragraph{Datasets}
We evaluate SHEAF and other measures on four benchmarks (Table~\ref{tab:datasets}). They span domains from
hand-crafted and deep-learning image features to encoded web text, intrinsic dimensionality
(mean LID) from $23$ to $40$, and ambient dimension from $100$ to $256$. We compute the exact
50-nearest-neighbor set for each query once on the GPU\@. Every experiment in this paper uses the
first 10 of these as ground truth ($k{=}10$), except one sensitivity setting in
\S\ref{sec:eval:sensitivity} that uses all 50 to test $k{=}50$.

\begin{table}[t]
  \centering
  \small
  \caption{Four evaluation datasets, ordered by mean LID.}
  \label{tab:datasets}
  \begin{tabular}{@{}llrrr@{}}
    \toprule
    Dataset & Domain & No. of Queries & Dimension & Mean LID \\
    \midrule
    SIFT~\cite{sift}                 & image & $10{,}000$ & $128$ & $23.4$ \\
    DEEP~\cite{deep}                 & image          & $1{,}000$  & $256$ & $29.9$ \\
    SPACEV~\cite{spacev}             & web text             & $29{,}316$ & $100$ & $30.0$ \\
    IMAGENET~\cite{imagenet}         & image                & $10{,}000$ & $100$ & $39.6$ \\
    \bottomrule
  \end{tabular}
\end{table}

\paragraph{Indexes and cost}
Our primary index is CAGRA with \texttt{nn\_descent}:
\texttt{graph\_degree} $64$, 
\texttt{intermediate\_graph\_degree} $128$, squared-L2. 
To test generality off the GPU,
we also build HNSW (hnswlib, $M{=}16$, \texttt{ef\_construction} $200$) on all four datasets. A query's cost
label is the smallest value on the ladder
$\{16,\allowbreak24,\allowbreak32,\allowbreak64,\allowbreak96,\allowbreak128,\allowbreak192,\allowbreak256,\allowbreak384,\allowbreak512\}$
(CAGRA's \texttt{itopk} or HNSW's \texttt{ef}) at which
its search reaches the recall target, evaluated at recall@$10$ $\ge \tau$ for
$\tau \in \{0.90, 0.95\}$. Queries that do not reach the target within the ladder are
right-censored (i.e., their cost are beyond the ladder scale) and excluded from the correlation. Under these settings, at $\tau{=}0.90$
the right-censored counts for SIFT, DEEP, SPACEV, and IMAGENET are
$0$, $0$, $127$, and $422$. At $\tau{=}0.95$, they are $40$, $1$, $670$, and $1{,}330$.

\paragraph{Baselines}
We compare six per-query hardness measures, laid out in Table~\ref{tab:signals}.
Three require ground truth: exact LID~\cite{amsaleg_lid} (the standard geometric proxy,
computed from true $k$-NN distances so its ceiling is measured, not an estimator's),
reverse-$k$-NN hubness~\cite{radovanovic_hubness} (R-$k$NN, the structural proxy), and
Steiner-hardness~\cite{steiner_hardness} (a recent graph-native measure, an offline
ground-truth-dependent reference). Three are computable during the search itself: the distance-based
probe (shallow-probe distances and their improvement, standing in for adaptive beam
search~\cite{adaptive_beam}), an online LID estimate (the Hill/MLE
estimator~\cite{houle_lid,amsaleg_lid} applied to the probe's own returned distances instead of
the true $k$-NN), and SHEAF proposed in this work.

\begin{table}[t]
  \centering
  \small
  \caption{Six hardness measures and availability at query time.}
  \label{tab:signals}
  \begin{tabular}{@{}llcc@{}}
    \toprule
    Measure & Complexity & Ground Truth & Query-time \\
    \midrule
    Exact LID~\cite{amsaleg_lid}                 & true $k$-NN   & needed & \xmark \\
    Hubness (R-$k$NN)~\cite{radovanovic_hubness} & base R-$k$NN  & needed & \xmark \\
    Steiner-hardness~\cite{steiner_hardness}     & NP-hard & needed & \xmark \\
    Distance-based probe~\cite{adaptive_beam}    & one probe     & no need   & \cmark \\
    Online LID estimate~\cite{amsaleg_lid}       & one probe     & no need   & \cmark \\
    SHEAF (this work)                       & one probe     & no need   & \cmark \\
    \bottomrule
  \end{tabular}

\end{table}

\paragraph{Metrics.}
Every measure is scored by how well it predicts a query's cost out of sample. The primary
metric is the held-out Pearson correlation with $\log$ of the cost label, fit on a random half
of the queries and reported on the other half. We summarize it as the flux-over-LID ratio of
these correlations.

\paragraph{Platform.}
All experiments run on a single Chameleon testbed node with two AMD EPYC 7763 $64$-core
CPUs ($256$ threads, $503$\,GB RAM) and one NVIDIA A100 80GB PCIe GPU (driver $560.35$,
CUDA $12.6$), under Ubuntu $24.04$ (kernel $6.8$). The software stack is Python $3.11$, cuVS
$26.06$ with CuPy $14.1$ for CAGRA and the exact GPU $k$-NN, hnswlib $0.8.0$ for HNSW, and
NumPy $2.4$. Figures use Matplotlib $3.11$. CAGRA and ground-truth $k$-NN run on the GPU, and
HNSW runs on the host CPU.

\subsection{Static Hardness Proxies}
\label{sec:eval:static}
Exact LID predicts the per-query beam only weakly on real CAGRA.
Figure~\ref{fig:static}(a) shows, per dataset, the held-out correlation of exact LID with the
per-query minimum \texttt{itopk}: it ranges from $0.42$ (SIFT) to $0.64$ (IMAGENET), never
strong, even though the beam itself varies several-fold across datasets
(mean minimum \texttt{itopk} from $28$ to $118$).

Static graph structure does no better than geometry. On SIFT1M
(Figure~\ref{fig:static}(b)),
the in-degree of a query's true neighbors, its anti-hubness, and
its BFS hop distance from the medoid all correlate with cost \emph{more weakly} than LID
(Pearson $\le 0.25$ vs.\ $0.33$). Adding every structural feature to a held-out regression on
top of LID lifts it only from $0.42$ to $0.47$. CAGRA balances node degree at construction,
which suppresses the hubness skew on which these features rely. This weakness is not a SIFT1M
artifact: held-out, hubness alone predicts the per-query beam with correlation $0.15$ to $0.30$
across the four datasets, below exact LID on every one of them
(Figure~\ref{fig:baselines}(a)).

\begin{figure}[t]
  \centering
  \includegraphics[width=0.8\columnwidth]{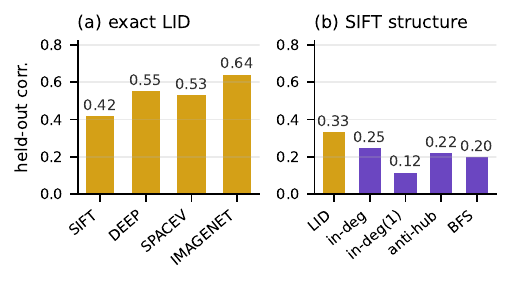}
  \caption{Static proxies predict the beam only weakly.}
  \label{fig:static}
\end{figure}

We also compare against Steiner-hardness on the one dataset where its NP-hard solve is tractable.
Steiner-hardness~\cite{steiner_hardness} scores a query by the minimum cost to answer it on a
separate reference graph built by its authors' own method, not CAGRA, using the query's
ground-truth neighbors and an NP-hard Directed-Steiner-Tree solve (authors' implementation,
SIFT1M). At recall $0.90$, Steiner-hardness correlates with the per-query beam more strongly
than exact LID ($0.47$ vs.\ $0.34$ held out), capturing real structure that geometry alone
misses. Answer-set flux (SHEAF) scores higher still ($0.53$), computed online without ground truth from the same
fixed narrow probe used throughout this paper (Table~\ref{tab:steiner}). 

\begin{table}[t]
  \centering
  \small
  \caption{On SIFT1M at recall $0.90$, flux scores higher than both the offline Steiner-hardness
  reference and exact LID.}
  \label{tab:steiner}
  \input{tab/tab_steiner}
\end{table}

\subsection{Flux as a Cost Predictor}
\label{sec:eval:flux}
 
Answer-set flux (SHEAF) predicts the per-query beam better than exact LID on every one of the four
datasets and at both recall targets (Figure~\ref{fig:flux-eval}). The
flux-over-LID ratio in held-out correlation ranges from $1.04\times$ to $1.55\times$. Flux is
never worse than LID\@. This gain holds even though flux's own probe stays strictly
narrower than the widths that hard queries ultimately need. These margins are stable under
resampling: across $50$ random half/half splits per dataset, flux's held-out correlation exceeds
exact LID's on every dataset and both targets by $4.2$ to $10.5$ standard errors (paired gap
$0.04$ to $0.16$). The advantage is therefore not an artifact of a single split.

Exact LID, though, is not the comparison that matters for a running system. It is given the
true $k$-NN distances of every query. No deployed system has this quantity. We include exact
LID because it is the field's standard difficulty proxy. Outperforming even this
generously-equipped baseline is a sanity check rather than the main objective. The measures that a running system could compute
online are the distance-based probe and an online LID estimate. These are the baselines that
matter most for a real comparison. 
The remainder of this subsection will compare SHEAF against these two online measures, and then show that the choice of regressor does not drive the result.

The online LID estimate applies the Hill/MLE estimator to the probe's own returned distances
rather than to true neighbors. 
By contrast, flux outperforms it on all four
datasets by $1.25$ to $2.0\times$ in held-out correlation (Figure~\ref{fig:baselines}(a)). This
margin is wider than flux's lead over exact LID\@. The probe's candidate set is a biased,
still-converging sample, and the Hill estimator inherits that bias. This leaves the online
estimate weaker than even exact LID (on SIFT1M, $0.27$ vs.\ $0.34$ at recall $0.90$ and $0.30$
vs.\ $0.42$ at recall $0.95$). The distance-based probe is the churn-free online quantity that an
adaptive beam search reads. Flux outperforms it on all four datasets by $1.05$ to $1.54\times$.
This margin isolates exactly what the top-$k$ membership
churn adds over the probe's distances alone: it is largest on the low-dimensional datasets
($1.54\times$ on SIFT) and thinnest on IMAGENET ($1.05\times$), whose high intrinsic dimension
leaves the probe's distances already carrying most of the information (more to be discussed in~\S\ref{sec:eval:ablation}).
Churn is thus a consistent but dataset-dependent gain over a distance-only probe.

The choice of regressor does not drive this result. We refit $h$ with a $k$-nearest-neighbor
regressor in place of the ordinary least-squares fit used everywhere else in the paper, which is a
nonparametric alternative from a different model family entirely. Flux still outperforms exact LID on
all four datasets and both recall targets under this alternative regressor
(Figure~\ref{fig:baselines}(b)).

\begin{figure}[t]
  \centering
  \includegraphics[width=\columnwidth]{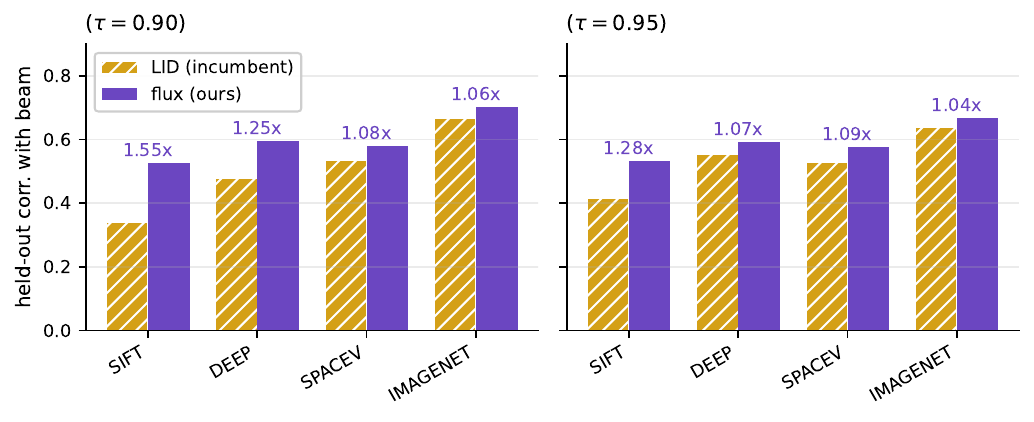}
  \caption{Flux (purple) outperforms LID (gold) on all four datasets and both recall targets; labels
  show the ratio.}
  \label{fig:flux-eval}
\end{figure}

\begin{figure}[t]
  \centering
  \includegraphics[width=\columnwidth]{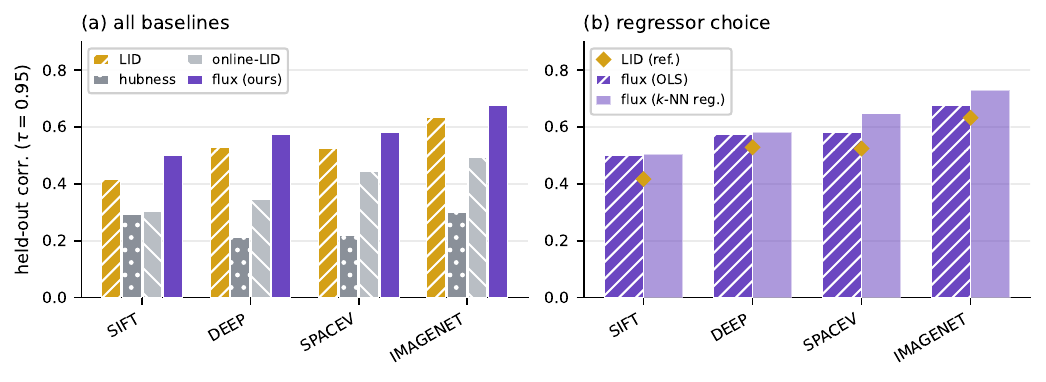}
  \caption{(a) Flux (SHEAF) outperforms LID, hubness, and online LID at recall $0.95$; (b) The same under a $k$-NN
  regressor.}
  \label{fig:baselines}
\end{figure}

\subsection{Generalization Across Index Architectures}
\label{sec:eval:generalize}
Flux (SHEAF) outperforming LID is a property of graph-based search, not of CAGRA or of GPUs. To separate the
two, we repeat the head-to-head on HNSW, a CPU graph index whose construction, unlike CAGRA's,
does not balance node degree. We use HNSW's \texttt{ef} in place of \texttt{itopk} and otherwise
the identical protocol: same $k$, exact LID, the same fixed probe ($w_0{=}16$, $w_1{=}24$),
beam ladder, and so forth.

Flux outperforms LID on HNSW on all four datasets at recall target $0.95$
(Figure~\ref{fig:hnsw}), by $1.06$ to $1.13\times$ in held-out correlation.
To verify that this phenomenon is extensible beyond a single target recall, we also compare the two measures at both recall $0.90$ and recall $0.95$.
Figure~\ref{fig:crossarch} shows the held-out correlation of flux and LID on both CAGRA and HNSW, at both recall targets. Flux outperforms LID on both architectures and both targets.
This is consistent with LID retaining more relevance on a graph whose degree is imbalanced.

\begin{figure}[t]
  \centering
  \includegraphics[width=\columnwidth]{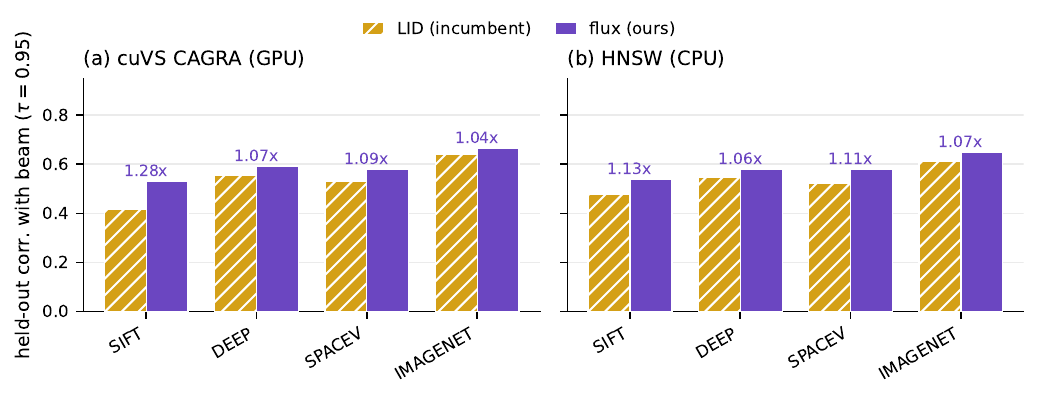}
  \caption{At recall $0.95$, all four datasets, on (a) CAGRA and (b) HNSW; labels show the
  flux-over-LID ratio.}
  \label{fig:hnsw}
\end{figure}

\begin{figure}[t]
  \centering
  \includegraphics[width=0.8\columnwidth]{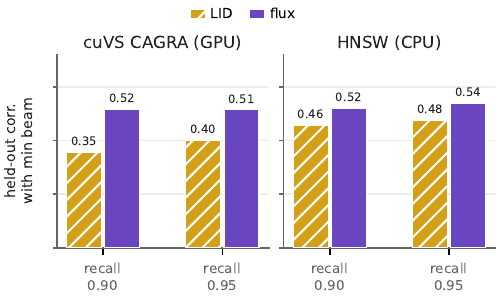}
  \caption{Held-out correlation on CAGRA and HNSW.}
  \label{fig:crossarch}
\end{figure}

\subsection{Parameter Sensitivity}
\label{sec:eval:sensitivity}

On all datasets, we vary seven settings, each changing one knob off the default probe: s1
is the default itself ($w{=}(16,24)$, $\tau{=}0.95$, $k{=}10$, full ladder); s2 widens the probe
to $w{=}(16,32)$; s3 to $w{=}(24,32)$; s4 to $w{=}(16,48)$; s5 lowers the recall target to
$\tau{=}0.90$; s6 raises $k$ to $50$ with $w{=}(64,96)$; and s7 coarsens the \texttt{itopk}
ladder. Across all seven settings (Figure~\ref{fig:sensitivity}), flux (SHEAF) predicts the per-query
beam better than exact LID\@. The ratio ranges from $1.05\times$ to $2.03\times$ and never drops
to parity. The flux-over-LID advantage is therefore stable across every probe and evaluation
knob varied here, not an artifact of a specific setting on any one dataset.

The default setting is not selected for the largest advantage. Widening the probe's second step
to $w{=}(16,48)$ or grading recall at $k{=}50$ increases the margin on all four datasets. The
default $w{=}(16,24)$ keeps the probe at the low end of the cost ladder.

\begin{figure}[t]
  \centering
  \includegraphics[width=\columnwidth]{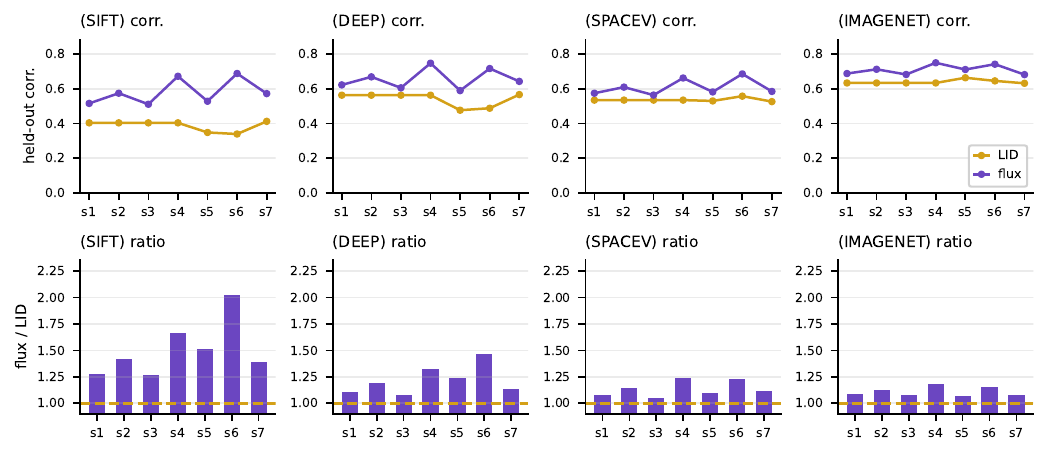}
  \caption{Flux/SHEAF (purple) outperforms LID (gold) across seven settings (top); the ratio stays above parity (bottom).}
  \label{fig:sensitivity}
\end{figure}

\subsection{Ablation Study}
\label{sec:eval:ablation}

Answer-set churn is the component that carries most of flux's prediction. On its own, churn
already predicts the per-query beam about as well as exact LID, and better on the
lower-dimensional datasets (Figure~\ref{fig:ablation}). It dominates the $k$-th distance
improvement in every case. The probe's raw distance profile adds a complementary increment,
largest on the highest-dimensional dataset (IMAGENET). The full flux predictor is the strongest
on all four datasets.

The churn term requires the shallow probe to overlap the operating beam, so that widening the
probe moves some queries' top-$k$ before the target is reached. All four datasets satisfy this:
their per-query beams sit within reach of the $16\to24$ probe, which makes churn
well-defined and informative across the panel.

\begin{figure}[t]
  \centering
  \includegraphics[width=0.9\columnwidth]{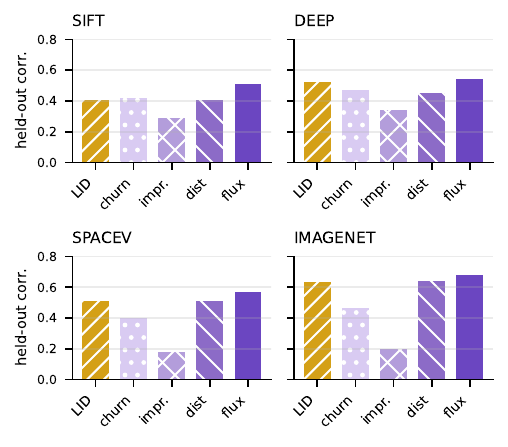}
  \caption{Churn leads the flux components; full flux/SHEAF (purple) is best on all four datasets.}
  \label{fig:ablation}
\end{figure}

\section{Conclusion}
\label{sec:conclusion}

This paper presented SHEAF, a self-profiled predictor of per-query cost in graph-based
approximate nearest neighbor search. Answer-set flux reads the churn of a query's own top-$k$
result across a shallow probe. This churn provably bounds the recall that a query can still gain
over that probe interval. A self-profiling estimator turns flux into a per-query beam prediction
at a fixed, index-agnostic probe cost. A fixed-probe evaluation protocol scores each measure
over all answerable queries and provides a pre-target guarantee on the hard tail.
Our extensive evaluation spans four domain-diverse datasets, five baseline measures including
Steiner-hardness, and two architectures, namely CAGRA on GPU and HNSW on CPU\@. Across all of
them, flux predicts the per-query beam better than exact local intrinsic dimensionality. SHEAF thus turns query hardness
from a static geometric guess into a cheap, online quantity that an adaptive per-query serving
policy can act on directly.

The directions that SHEAF opens concern the hardness measure rather than any specific system. A first direction is to enrich the measure itself: a probe exposes more of a query's early behavior than the top-$k$ churn and distance improvement that flux reads. Any such richer feature could then be scored under the same fixed-probe protocol without new evaluation machinery. A second direction extends validation beyond the four datasets here to further graph architectures, larger workloads, and production scale. Finally, the predictor could drive an adaptive serving policy that gives each query its own beam width: small for the many easy queries, large only for the hard tail. Because every query first runs the probe, such a policy is worthwhile only when the beam that it saves on the easy queries outweighs that probe cost across the workload.


\begin{acks}
  Results presented in this paper were obtained using the Chameleon testbed supported by the National Science Foundation.

\paragraph{Disclosure of AI Usage}
Source code and manuscript editing were assisted by Copilot and Claude. All key components of
this work, including the design and analysis of the flux predictor, the evaluation methodology,
and the interpretation of results, were developed and executed by the authors. The authors have
verified the accuracy and integrity of the content to the best of their knowledge and take full
responsibility for the final manuscript.
\end{acks}

\bibliographystyle{ACM-Reference-Format}
\bibliography{refs}

\clearpage

\appendix

\section{Appendix: Proofs and Additional Remarks}
\label{app:proofs}

\subsection{Churn bounds recall change over probe interval}
\label{app:churn}
\textit{Proposition~\ref{prop:churn}. For any query and widths $w_0<w_1$ with
$|R_0|=|R_1|=k$, $\bigl|\rho_{w_1}(q)-\rho_{w_0}(q)\bigr|\le|R_0\triangle R_1|/(2k)=\chi(q)$.}

\begin{proof}
Fix a query $q$ and widths $w_0<w_1$, and write $R_0,R_1$ for the returned top-$k$ sets with
$|R_0|=|R_1|=k$, $N$ for the true $k$-NN set, and $\rho_i=|R_i\cap N|/k$. Let
$a=|R_0\cap R_1|$. Since both sets have size $k$,
\[
  |R_1\setminus R_0|=|R_0\setminus R_1|=k-a,\qquad
  |R_0\triangle R_1|=2(k-a).
\]
Split the intersections with $N$ along the common part,
\[
  |R_1\cap N|-|R_0\cap N|
  =\bigl|(R_1\setminus R_0)\cap N\bigr|-\bigl|(R_0\setminus R_1)\cap N\bigr|,
\]
because the shared elements $R_0\cap R_1$ contribute the same count to each side and cancel.
Each term on the right lies in $[0,\,k-a]$; their difference is therefore at most $k-a$ in
absolute value:
\[
  \bigl||R_1\cap N|-|R_0\cap N|\bigr|\le k-a=\tfrac12|R_0\triangle R_1|.
\]
Dividing by $k$ yields $|\rho_{w_1}-\rho_{w_0}|\le|R_0\triangle R_1|/(2k)=\chi(q)$. Equality
holds exactly when every element that enters the answer is a true neighbor and every element
that leaves is not (or vice versa), i.e., when all churn is aligned with $N$. In particular
$\chi(q)=0\Rightarrow\rho_{w_1}=\rho_{w_0}$: a stable answer set certifies that recall has
plateaued across the probe.
\end{proof}

\subsection{Fixed-cost, index-agnostic estimation}
\label{app:overhead}
\textit{Proposition~\ref{prop:overhead}. With $W(w)$ the per-query search work at width $w$,
nondecreasing in $w$: (i) the
estimator is well defined for any ANN index exposing a width-parameterized top-$k$, regardless
of $W$; and (ii) if in addition $W(w)=\Theta(w)$, a linear-cost condition that both CAGRA's
\texttt{itopk} search and HNSW's \texttt{ef} search satisfy, its probe cost is
$W(w_1)=\Theta(w_1)$ per query, independent
of $c(q)$. The relative overhead $W(w_1)/W(c(q))=\Theta(w_1/c(q))$ therefore tends to $0$ on the
hard tail $c(q)\gg w_1$.}

\begin{proof}
Write $W(w)$ for the per-query work of a width-$w$ search, measured in node-distance
evaluations. Claim (i) holds for any nondecreasing $W$. Claim (ii) additionally assumes a
bounded-candidate-list search model: the search maintains an internal candidate buffer of size
$\Theta(w)$; each step expands one candidate to its fixed out-degree neighbors, an $\Theta(1)$
operation. Under this model $W(w)=\Theta(w)$. CAGRA's \texttt{itopk} search
satisfies this model directly: the candidate buffer has size $\Theta(w)$ and each iteration does
$\Theta(1)$ expansions. HNSW's \texttt{ef} search has the same structure: the dynamic candidate
list is bounded by $\mathrm{ef}=w$, and each visited node expands to a fixed out-degree $M$ (a
construction-time constant, $M{=}16$ in our setup); its work is therefore likewise $\Theta(w)$. This is
a structural argument about the search procedure, not a measured timing curve; we do not report
wall-clock or node-visit counts against $w$ for either index.

For (i), recall that the feature map $\phi(q)$ in Eq.~\eqref{eq:feat} depends only on $R_{w_0}(q),R_{w_1}(q)$
and their $k$-th/mean returned distances. Every one of these is part of the output of a
top-$k$ search at a specified width, which any width-parameterized ANN index exposes; no
graph internals, distance cache, or index-specific state is read. Hence the estimator is well
defined on any such index, and the identical $\phi$ and predictor $h$ apply to CAGRA and to
HNSW.

For (ii), the probe runs two searches, at $w_0$ and $w_1$, costing
$W(w_0)+W(w_1)=\Theta(w_1)$ since $w_0<w_1$, while forming $\phi$ and evaluating $h$ is $O(1)$
arithmetic. This budget is a constant in the query's true difficulty. A query of cost $c(q)$ served at its
required width pays $W(c(q))=\Theta(c(q))$. The relative overhead is therefore
$W(w_1)/W(c(q))=\Theta(w_1/c(q))$, which tends to $0$ as $c(q)\gg w_1$. That is precisely the
hard, expensive tail an adaptive scheme must budget for.
\end{proof}

\begin{remark}
This bound charges the probe in full. An implementation that resumes the wider search
from the probe's own state, rather than discarding it, could reuse part of that work. We do not
implement or measure this. The reported figures are therefore a conservative upper bound on the
probe's overhead.
\end{remark}

\subsection{The probe remains below the target}
\label{app:decirc}
\textit{Proposition~\ref{prop:decirc}. If $q$ is answerable and $w_1\in L$ with $w_1<c(q)$, then
$\rho_{w_1}(q)<\tau$.}

\begin{proof}
Fix a query $q$, a target $\tau$, and the ladder $L=\{b_1<\cdots<b_m\}$. Write
$A(q)=\{b\in L:\rho_b(q)\ge\tau\}$ for the set of ladder widths that attain the target.
Suppose $w_1\in L$ and $w_1<c(q)$.
Since $q$ is answerable, $A(q)$ is nonempty and $c(q)=\min A(q)$. No width smaller than $c(q)$
belongs to $A(q)$, so $w_1\notin A(q)$ and $\rho_{w_1}(q)<\tau$.
\end{proof}

\begin{remark}
The argument does not assume $\rho_b(q)$ is
nondecreasing in $b$. This distinction matters in practice: CAGRA's recall curve is not
guaranteed monotone in \texttt{itopk}, since changing the beam width alters the greedy traversal,
and a wider search need not visit or retain exactly the true neighbors that a narrower one recovered.
Proposition~\ref{prop:decirc} holds regardless, because it only
rules out $w_1$ lying in the attaining set $A(q)$ and says nothing about the shape of $\rho_b(q)$
between $0$ and $w_1$.
\end{remark}

\begin{remark}
This makes the pre-target guarantee meaningful. On any query with $w_1<c(q)$, that is,
the hard tail that the fixed probe cannot answer, the probe result $R_{w_1}$ has recall strictly
below $\tau$. The probe has not yet recovered enough true neighbors to answer the query at the
target confidence. A measure computed from $R_{w_1}$ therefore cannot be scored on a target that
the probe has already attained, since it has not. It must instead extrapolate from a result that
is still below target. This is the sense in which flux and the other measures are evaluated on
the hard tail: before the probe reaches the recall that it is used to predict, not after.
\end{remark}

\subsection{Extra work from too-wide beam bounded by ladder granularity}
\label{app:misprediction}
\textit{Proposition~\ref{prop:misprediction}. Let $L=\{b_1<\cdots<b_m\}$ and fix a query $q$
with $c(q)=b_k$. If the estimator serves rung $b_{k+j}$ for some $j\ge 0$, the excess work
relative to serving at $c(q)$ satisfies
$W(b_{k+j})-W(c(q))\le j\cdot\max_{1\le i<m}(W(b_{i+1})-W(b_i))$.}

\begin{proof}
Fix a query $q$ with $c(q)=b_k$, and suppose the estimator serves rung $b_{k+j}$ for some
integer $j\ge 0$ with $k+j\le m$. Write $\delta=\max_{1\le i<m}\bigl(W(b_{i+1})-W(b_i)\bigr)$ for
the largest work gap between two adjacent rungs of $L$. Since $W$ is nondecreasing, telescope
the excess work into $j$ consecutive rung-to-rung steps:
\[
  W(b_{k+j})-W(b_k)
  =\sum_{i=k}^{k+j-1}\bigl(W(b_{i+1})-W(b_i)\bigr)
  \le\sum_{i=k}^{k+j-1}\delta
  =j\,\delta.
\]
Each summand is nonnegative because $W$ is nondecreasing, and each is at most $\delta$ by the
definition of $\delta$ as the largest such gap on $L$; the sum of $j$ terms is therefore at most
$j\delta$. Substituting $c(q)=b_k$ gives $W(b_{k+j})-W(c(q))\le j\delta$, which is the claim.
\end{proof}

\begin{remark}
The integer $j$ counts how many rungs above $c(q)$ that the estimator served, that is, how much
wider than needed the served beam was. The bound depends only on this $j$ and on $\delta$, which is a
property of the ladder $L$ alone. It does not depend on the query, the dataset, or the
estimator's internals: any two estimators that serve a query $j$ rungs too wide pay
the same worst-case excess, regardless of how their predictions were produced. A coarser ladder
has fewer, more widely spaced rungs. This shrinks $j$ for a given prediction error measured in
absolute beam units, but it grows $\delta$; the bound makes this trade-off explicit rather than
leaving it implicit in the choice of $L$.
\end{remark}

\begin{remark}
The proposition is silent on a too-narrow beam by construction. An estimate that serves a rung below
$c(q)$ is not a work-cost event; by Proposition~\ref{prop:decirc} it is a recall event, since
that rung already fails to reach the target. The two failure modes are therefore addressed by
two different guarantees. Proposition~\ref{prop:decirc} certifies that any prediction below
$c(q)$ misses the target. This proposition bounds only the excess work of a prediction above
$c(q)$. Because recall need not be monotone in the beam width, this proposition does not bound
the recall deficit of a wider rung or guarantee that the wider rung attains the target; it
bounds the search work $W$, not recall.
\end{remark}

\end{document}

%% file: tab/tab_steiner.tex
\begin{tabularx}{\columnwidth}{@{}lXr@{}}
\toprule
Measure & availability & correlation \\
\midrule
exact LID & offline, true $k$-NN & 0.34 \\
Steiner-hardness & offline, ground-truth-dependent & 0.47 \\
\midrule
SHEAF (this work) & \textbf{online, ground truth not needed} & \textbf{0.53} \\
\bottomrule
\end{tabularx}

%% file: refs.bib
@misc{hua_hardness26,
  author       = {Hua, Zhiyuan and Mo, Qiji and Yao, Zebin and Cui, Lixiao and Liu, Xiaoguang and Wang, Gang and Wei, Zijing and Liu, Xinyu and Tang, Tianxiao and Liu, Shaozhi and Qu, Lin},
  title        = {Dynamically Detect and Fix Hardness for Efficient Approximate Nearest Neighbor Search},
  year         = {2025},
  eprint       = {2510.22316},
  archivePrefix= {arXiv},
  primaryClass = {cs.DB},
  note         = {To appear, SIGMOD 2026},
}

@misc{li_convergent26,
  author       = {Li, Binhong and Yan, Xiao and Lu, Shangqi},
  title        = {Fast-Convergent Proximity Graphs for Approximate Nearest Neighbor Search},
  year         = {2025},
  eprint       = {2510.05975},
  archivePrefix= {arXiv},
  primaryClass = {cs.DB},
  note         = {To appear, SIGMOD 2026},
}

@misc{khanna_navigable25,
  author       = {Khanna, Sanjeev and Padaki, Ashwin and Waingarten, Erik},
  title        = {Sparse Navigable Graphs for Nearest Neighbor Search: Algorithms and Hardness},
  year         = {2025},
  eprint        = {2507.14060},
  archivePrefix = {arXiv},
  primaryClass = {cs.DS},
}

@inproceedings{baranchuk_route19,
  author    = {Baranchuk, Dmitry and Persiyanov, Dmitry and Sinitsin, Anton and Babenko, Artem},
  title     = {Learning to Route in Similarity Graphs},
  booktitle = {Proceedings of the 36th International Conference on Machine Learning (ICML)},
  series    = {PMLR},
  volume    = {97},
  year      = {2019},
}

@inproceedings{prokhorenkova_theory20,
  author    = {Prokhorenkova, Liudmila and Shekhovtsov, Aleksandr},
  title      = {Graph-based Nearest Neighbor Search: From Practice to Theory},
  booktitle = {Proceedings of the 37th International Conference on Machine Learning (ICML)},
  series    = {PMLR},
  volume    = {119},
  year      = {2020},
}

@inproceedings{chen_spann21,
  author    = {Chen, Qi and Zhao, Bing and Wang, Haidong and Li, Mingqin and Liu, Chuanjie and Li, Zengzhong and Yang, Mao and Wang, Jingdong},
  title     = {{SPANN}: Highly-efficient Billion-scale Approximate Nearest Neighbor Search},
  booktitle = {Advances in Neural Information Processing Systems (NeurIPS)},
  year      = {2021},
}

@inproceedings{dzhao_hpdc26,
author = {Zhao, Dongfang},
title = {SIVF: GPU-Resident IVF Index for Streaming Vector Analytics},
year = {2026},
isbn = {9798400726408},
publisher = {Association for Computing Machinery},
address = {New York, NY, USA},
url = {https://doi.org/10.1145/3806645.3807575},
doi = {10.1145/3806645.3807575},
booktitle = {Proceedings of the 35th International Symposium on High-Performance Parallel and Distributed Computing},
pages = {31–44},
numpages = {14},
location = {
},
series = {HPDC '26}
}

@inproceedings{jwang_sigmod21,
author = {Wang, Jianguo and Yi, Xiaomeng and Guo, Rentong and Jin, Hai and Xu, Peng and Li, Shengjun and Wang, Xiangyu and Guo, Xiangzhou and Li, Chengming and Xu, Xiaohai and Yu, Kun and Yuan, Yuxing and Zou, Yinghao and Long, Jiquan and Cai, Yudong and Li, Zhenxiang and Zhang, Zhifeng and Mo, Yihua and Gu, Jun and Jiang, Ruiyi and Wei, Yi and Xie, Charles},
title = {Milvus: A Purpose-Built Vector Data Management System},
year = {2021},
isbn = {9781450383431},
publisher = {Association for Computing Machinery},
address = {New York, NY, USA},
url = {https://doi.org/10.1145/3448016.3457550},
doi = {10.1145/3448016.3457550},
booktitle = {Proceedings of the 2021 International Conference on Management of Data},
pages = {2614–2627},
numpages = {14},
location = {Virtual Event, China},
series = {SIGMOD '21}
}

@inproceedings{plewi_neurips20,
author = {Lewis, Patrick and Perez, Ethan and Piktus, Aleksandra and Petroni, Fabio and Karpukhin, Vladimir and Goyal, Naman and K\"{u}ttler, Heinrich and Lewis, Mike and Yih, Wen-tau and Rockt\"{a}schel, Tim and Riedel, Sebastian and Kiela, Douwe},
title = {Retrieval-augmented generation for knowledge-intensive NLP tasks},
year = {2020},
isbn = {9781713829546},
publisher = {Curran Associates Inc.},
address = {Red Hook, NY, USA},
booktitle = {Proceedings of the 34th International Conference on Neural Information Processing Systems},
articleno = {793},
numpages = {16},
location = {Vancouver, BC, Canada},
series = {NIPS '20}
}

@inproceedings{rying_kdd18,
author = {Ying, Rex and He, Ruining and Chen, Kaifeng and Eksombatchai, Pong and Hamilton, William L. and Leskovec, Jure},
title = {Graph Convolutional Neural Networks for Web-Scale Recommender Systems},
year = {2018},
isbn = {9781450355520},
publisher = {Association for Computing Machinery},
address = {New York, NY, USA},
url = {https://doi.org/10.1145/3219819.3219890},
doi = {10.1145/3219819.3219890},
booktitle = {Proceedings of the 24th ACM SIGKDD International Conference on Knowledge Discovery \& Data Mining},
pages = {974–983},
numpages = {10},
location = {London, United Kingdom},
series = {KDD '18}
}

@article{yzhu_sigmod24,
author = {Zhu, Yifan and Ma, Ruiyao and Zheng, Baihua and Ke, Xiangyu and Chen, Lu and Gao, Yunjun},
title = {GTS: GPU-based Tree Index for Fast Similarity Search},
year = {2024},
issue_date = {June 2024},
publisher = {Association for Computing Machinery},
address = {New York, NY, USA},
volume = {2},
number = {3},
url = {https://doi.org/10.1145/3654945},
doi = {10.1145/3654945},
journal = {Proc. ACM Manag. Data},
month = may,
articleno = {142},
numpages = {27}
}

@inproceedings{hwang_cikm21,
author = {Wang, Hui and Zhao, Wan-Lei and Zeng, Xiangxiang and Yang, Jianye},
title = {Fast k-NN Graph Construction by GPU based NN-Descent},
year = {2021},
isbn = {9781450384469},
publisher = {Association for Computing Machinery},
address = {New York, NY, USA},
url = {https://doi.org/10.1145/3459637.3482344},
doi = {10.1145/3459637.3482344},
booktitle = {Proceedings of the 30th ACM International Conference on Information \& Knowledge Management},
pages = {1929–1938},
numpages = {10},
location = {Virtual Event, Queensland, Australia},
series = {CIKM '21}
}

@article{zli_sigmod25,
author = {Li, Zhonggen and Ke, Xiangyu and Zhu, Yifan and Yu, Bocheng and Zheng, Baihua and Gao, Yunjun},
title = {Scalable Graph Indexing using GPUs for Approximate Nearest Neighbor Search},
year = {2025},
issue_date = {December 2025},
publisher = {Association for Computing Machinery},
address = {New York, NY, USA},
volume = {3},
number = {6},
url = {https://doi.org/10.1145/3769825},
doi = {10.1145/3769825},
journal = {Proc. ACM Manag. Data},
month = dec,
articleno = {360},
numpages = {27}
}

@ARTICLE{kvenk_tbd25,
  author={Venkatasubba, Karthik and Khan, Saim and Singh, Somesh and Simhadri, Harsha Vardhan and Vedurada, Jyothi},
  journal={IEEE Transactions on Big Data}, 
  title={BANG: Billion-Scale Approximate Nearest Neighbour Search Using a Single GPU}, 
  year={2025},
  volume={11},
  number={6},
  pages={3142-3157},
  doi={10.1109/TBDATA.2025.3581085}}

@inproceedings{wzhao_icde20,
  author       = {Weijie Zhao and
                  Shulong Tan and
                  Ping Li},
  title        = {{SONG:} Approximate Nearest Neighbor Search on {GPU}},
  booktitle    = {36th {IEEE} International Conference on Data Engineering, {ICDE} 2020,
                  Dallas, TX, USA, April 20-24, 2020},
  pages        = {1033--1044},
  publisher    = {{IEEE}},
  year         = {2020},
  url          = {https://doi.org/10.1109/ICDE48307.2020.00094},
  doi          = {10.1109/ICDE48307.2020.00094},
  bibsource    = {dblp computer science bibliography, https://dblp.org}
}

@article{mwang_sigmod24,
author = {Wang, Mengzhao and Xu, Weizhi and Yi, Xiaomeng and Wu, Songlin and Peng, Zhangyang and Ke, Xiangyu and Gao, Yunjun and Xu, Xiaoliang and Guo, Rentong and Xie, Charles},
title = {Starling: An I/O-Efficient Disk-Resident Graph Index Framework for High-Dimensional Vector Similarity Search on Data Segment},
year = {2024},
issue_date = {February 2024},
publisher = {Association for Computing Machinery},
address = {New York, NY, USA},
volume = {2},
number = {1},
url = {https://doi.org/10.1145/3639269},
doi = {10.1145/3639269},
journal = {Proc. ACM Manag. Data},
month = mar,
articleno = {14},
numpages = {27}
}

@article{cwei_vldb20,
author = {Wei, Chuangxian and Wu, Bin and Wang, Sheng and Lou, Renjie and Zhan, Chaoqun and Li, Feifei and Cai, Yuanzhe},
title = {AnalyticDB-V: a hybrid analytical engine towards query fusion for structured and unstructured data},
year = {2020},
issue_date = {August 2020},
publisher = {VLDB Endowment},
volume = {13},
number = {12},
issn = {2150-8097},
url = {https://doi.org/10.14778/3415478.3415541},
doi = {10.14778/3415478.3415541},
journal = {Proc. VLDB Endow.},
month = aug,
pages = {3152–3165},
numpages = {14}
}

@inproceedings{qzhan_osdi23,
author = {Qianxi Zhang and Shuotao Xu and Qi Chen and Guoxin Sui and Jiadong Xie and Zhizhen Cai and Yaoqi Chen and Yinxuan He and Yuqing Yang and Fan Yang and Mao Yang and Lidong Zhou},
title = {{VBASE}: Unifying Online Vector Similarity Search and Relational Queries via Relaxed Monotonicity},
booktitle = {17th USENIX Symposium on Operating Systems Design and Implementation (OSDI 23)},
year = {2023},
isbn = {978-1-939133-34-2},
address = {Boston, MA},
pages = {377--395},
url = {https://www.usenix.org/conference/osdi23/presentation/zhang-qianxi},
publisher = {USENIX Association},
month = jul
}

@article{lpate_sigmod24,
author = {Patel, Liana and Kraft, Peter and Guestrin, Carlos and Zaharia, Matei},
title = {ACORN: Performant and Predicate-Agnostic Search Over Vector Embeddings and Structured Data},
year = {2024},
issue_date = {June 2024},
publisher = {Association for Computing Machinery},
address = {New York, NY, USA},
volume = {2},
number = {3},
url = {https://doi.org/10.1145/3654923},
doi = {10.1145/3654923},
journal = {Proc. ACM Manag. Data},
month = may,
articleno = {120},
numpages = {27}
}

@inproceedings{sgoll_www23,
author = {Gollapudi, Siddharth and Karia, Neel and Sivashankar, Varun and Krishnaswamy, Ravishankar and Begwani, Nikit and Raz, Swapnil and Lin, Yiyong and Zhang, Yin and Mahapatro, Neelam and Srinivasan, Premkumar and Singh, Amit and Simhadri, Harsha Vardhan},
title = {Filtered-DiskANN: Graph Algorithms for Approximate Nearest Neighbor Search with Filters},
year = {2023},
isbn = {9781450394161},
publisher = {Association for Computing Machinery},
address = {New York, NY, USA},
url = {https://doi.org/10.1145/3543507.3583552},
doi = {10.1145/3543507.3583552},
booktitle = {Proceedings of the ACM Web Conference 2023},
pages = {3406–3416},
numpages = {11},
location = {Austin, TX, USA},
series = {WWW '23}
}

@inproceedings{jruan_kdd25,
author = {Ruan, Jiancheng and Chen, Tingyang and Yang, Renchi and Ke, Xiangyu and Gao, Yunjun},
title = {Empowering Graph-based Approximate Nearest Neighbor Search with Adaptive Awareness Capabilities},
year = {2025},
isbn = {9798400714542},
publisher = {Association for Computing Machinery},
address = {New York, NY, USA},
url = {https://doi.org/10.1145/3711896.3736930},
doi = {10.1145/3711896.3736930},
booktitle = {Proceedings of the 31st ACM SIGKDD Conference on Knowledge Discovery and Data Mining V.2},
pages = {2444–2454},
numpages = {11},
location = {Toronto ON, Canada},
series = {KDD '25}
}

@inproceedings{ypan_bigdata23,
  author       = {Yu Pan and
                  Jianxin Sun and
                  Hongfeng Yu},
  editor       = {Jingrui He and
                  Themis Palpanas and
                  Xiaohua Hu and
                  Alfredo Cuzzocrea and
                  Dejing Dou and
                  Dominik Slezak and
                  Wei Wang and
                  Aleksandra Gruca and
                  Jerry Chun{-}Wei Lin and
                  Rakesh Agrawal},
  title        = {LM-DiskANN: Low Memory Footprint in Disk-Native Dynamic Graph-Based
                  {ANN} Indexing},
  booktitle    = {{IEEE} International Conference on Big Data, BigData 2023, Sorrento,
                  Italy, December 15-18, 2023},
  pages        = {5987--5996},
  publisher    = {{IEEE}},
  year         = {2023},
  url          = {https://doi.org/10.1109/BigData59044.2023.10386517},
  doi          = {10.1109/BIGDATA59044.2023.10386517},
  bibsource    = {dblp computer science bibliography, https://dblp.org}
}

@inproceedings{yxu_sosp23,
author = {Xu, Yuming and Liang, Hengyu and Li, Jin and Xu, Shuotao and Chen, Qi and Zhang, Qianxi and Li, Cheng and Yang, Ziyue and Yang, Fan and Yang, Yuqing and Cheng, Peng and Yang, Mao},
title = {SPFresh: Incremental In-Place Update for Billion-Scale Vector Search},
year = {2023},
isbn = {9798400702297},
publisher = {Association for Computing Machinery},
address = {New York, NY, USA},
url = {https://doi.org/10.1145/3600006.3613166},
doi = {10.1145/3600006.3613166},
booktitle = {Proceedings of the 29th Symposium on Operating Systems Principles},
pages = {545–561},
numpages = {17},
location = {Koblenz, Germany},
series = {SOSP '23}
}

@misc{asing_arxiv21,
      title={FreshDiskANN: A Fast and Accurate Graph-Based ANN Index for Streaming Similarity Search}, 
      author={Aditi Singh and Suhas Jayaram Subramanya and Ravishankar Krishnaswamy and Harsha Vardhan Simhadri},
      year={2021},
      eprint={2105.09613},
      archivePrefix={arXiv},
      primaryClass={cs.IR},
      url={https://arxiv.org/abs/2105.09613}, 
}

@misc{mdouz_arxiv24,
      title={The Faiss library}, 
      author={Matthijs Douze and Alexandr Guzhva and Chengqi Deng and Jeff Johnson and Gergely Szilvasy and Pierre-Emmanuel Mazaré and Maria Lomeli and Lucas Hosseini and Hervé Jégou},
      year={2025},
      eprint={2401.08281},
      archivePrefix={arXiv},
      primaryClass={cs.LG},
      url={https://arxiv.org/abs/2401.08281}, 
}

@InProceedings{jmart_eccv18,
author = {Martinez, Julieta and Zakhmi, Shobhit and Hoos, Holger H. and Little, James J.},
title = {LSQ++: Lower running time and higher recall in multi-codebook quantization},
booktitle = {Proceedings of the European Conference on Computer Vision (ECCV)},
month = {September},
year = {2018}
}

@InProceedings{rguo_icml20,
  title = 	 {Accelerating Large-Scale Inference with Anisotropic Vector Quantization},
  author =       {Guo, Ruiqi and Sun, Philip and Lindgren, Erik and Geng, Quan and Simcha, David and Chern, Felix and Kumar, Sanjiv},
  booktitle = 	 {Proceedings of the 37th International Conference on Machine Learning},
  pages = 	 {3887--3896},
  year = 	 {2020},
  editor = 	 {III, Hal Daumé and Singh, Aarti},
  volume = 	 {119},
  series = 	 {Proceedings of Machine Learning Research},
  month = 	 {13--18 Jul},
  publisher =    {PMLR},
  url = 	 {https://proceedings.mlr.press/v119/guo20h.html},
}

@article{jgao_sigmod24,
author = {Gao, Jianyang and Long, Cheng},
title = {RaBitQ: Quantizing High-Dimensional Vectors with a Theoretical Error Bound for Approximate Nearest Neighbor Search},
year = {2024},
issue_date = {June 2024},
publisher = {Association for Computing Machinery},
address = {New York, NY, USA},
volume = {2},
number = {3},
url = {https://doi.org/10.1145/3654970},
doi = {10.1145/3654970},
journal = {Proc. ACM Manag. Data},
month = may,
articleno = {167},
numpages = {27}
}

@article{tge_tpami13,
author = {Ge, Tiezheng and He, Kaiming and Ke, Qifa and Sun, Jian},
title = {Optimized Product Quantization},
year = {2014},
issue_date = {April 2014},
publisher = {IEEE Computer Society},
address = {USA},
volume = {36},
number = {4},
issn = {0162-8828},
url = {https://doi.org/10.1109/TPAMI.2013.240},
doi = {10.1109/TPAMI.2013.240},
journal = {IEEE Trans. Pattern Anal. Mach. Intell.},
month = apr,
pages = {744–755},
numpages = {12}
}

@ARTICLE{mmuja_tpami14,
author={Muja, Marius and Lowe, David G.},
journal={ IEEE Transactions on Pattern Analysis \& Machine Intelligence },
title={{ Scalable Nearest Neighbor Algorithms for High Dimensional Data }},
year={2014},
volume={36},
number={11},
ISSN={1939-3539},
pages={2227-2240},
doi={10.1109/TPAMI.2014.2321376},
url = {https://doi.ieeecomputersociety.org/10.1109/TPAMI.2014.2321376},
publisher={IEEE Computer Society},
address={Los Alamitos, CA, USA},
month=nov}

@article{jbent_cacm75,
author = {Bentley, Jon Louis},
title = {Multidimensional binary search trees used for associative searching},
year = {1975},
issue_date = {Sept. 1975},
publisher = {Association for Computing Machinery},
address = {New York, NY, USA},
volume = {18},
number = {9},
issn = {0001-0782},
url = {https://doi.org/10.1145/361002.361007},
doi = {10.1145/361002.361007},
journal = {Commun. ACM},
month = sep,
pages = {509–517},
numpages = {9}
}

@inproceedings{aando_neurips15,
 author = {Andoni, Alexandr and Indyk, Piotr and Laarhoven, Thijs and Razenshteyn, Ilya and Schmidt, Ludwig},
 booktitle = {Advances in Neural Information Processing Systems},
 editor = {C. Cortes and N. Lawrence and D. Lee and M. Sugiyama and R. Garnett},
 pages = {},
 publisher = {Curran Associates, Inc.},
 title = {Practical and Optimal LSH for Angular Distance},
 url = {https://proceedings.neurips.cc/paper_files/paper/2015/file/2823f4797102ce1a1aec05359cc16dd9-Paper.pdf},
 volume = {28},
 year = {2015}
}

@inproceedings{qlv_vldb07,
author = {Lv, Qin and Josephson, William and Wang, Zhe and Charikar, Moses and Li, Kai},
title = {Multi-probe LSH: efficient indexing for high-dimensional similarity search},
year = {2007},
isbn = {9781595936493},
publisher = {VLDB Endowment},
booktitle = {Proceedings of the 33rd International Conference on Very Large Data Bases},
pages = {950–961},
numpages = {12},
location = {Vienna, Austria},
series = {VLDB '07}
}

@inproceedings{mdata_socg04,
author = {Datar, Mayur and Immorlica, Nicole and Indyk, Piotr and Mirrokni, Vahab S.},
title = {Locality-sensitive hashing scheme based on p-stable distributions},
year = {2004},
isbn = {1581138857},
publisher = {Association for Computing Machinery},
address = {New York, NY, USA},
url = {https://doi.org/10.1145/997817.997857},
doi = {10.1145/997817.997857},
booktitle = {Proceedings of the Twentieth Annual Symposium on Computational Geometry},
pages = {253–262},
numpages = {10},
location = {Brooklyn, New York, USA},
series = {SCG '04}
}

@inproceedings{pindy_stoc98,
author = {Indyk, Piotr and Motwani, Rajeev},
title = {Approximate nearest neighbors: towards removing the curse of dimensionality},
year = {1998},
isbn = {0897919629},
publisher = {Association for Computing Machinery},
address = {New York, NY, USA},
url = {https://doi.org/10.1145/276698.276876},
doi = {10.1145/276698.276876},
booktitle = {Proceedings of the Thirtieth Annual ACM Symposium on Theory of Computing},
pages = {604–613},
numpages = {10},
location = {Dallas, Texas, USA},
series = {STOC '98}
}

@inproceedings{mmano_ppopp24,
author = {Manohar, Magdalen Dobson and Shen, Zheqi and Blelloch, Guy and Dhulipala, Laxman and Gu, Yan and Simhadri, Harsha Vardhan and Sun, Yihan},
title = {ParlayANN: Scalable and Deterministic Parallel Graph-Based Approximate Nearest Neighbor Search Algorithms},
year = {2024},
isbn = {9798400704352},
publisher = {Association for Computing Machinery},
address = {New York, NY, USA},
url = {https://doi.org/10.1145/3627535.3638475},
doi = {10.1145/3627535.3638475},
booktitle = {Proceedings of the 29th ACM SIGPLAN Annual Symposium on Principles and Practice of Parallel Programming},
pages = {270–285},
numpages = {16},
location = {Edinburgh, United Kingdom},
series = {PPoPP '24}
}

@inproceedings{nono_mm23,
author = {Ono, Naoki and Matsui, Yusuke},
title = {Relative NN-Descent: A Fast Index Construction for Graph-Based Approximate Nearest Neighbor Search},
year = {2023},
isbn = {9798400701085},
publisher = {Association for Computing Machinery},
address = {New York, NY, USA},
url = {https://doi.org/10.1145/3581783.3612290},
doi = {10.1145/3581783.3612290},
booktitle = {Proceedings of the 31st ACM International Conference on Multimedia},
pages = {1659–1667},
numpages = {9},
location = {Ottawa ON, Canada},
series = {MM '23}
}

@misc{cfu_arxiv16,
      title={EFANNA : An Extremely Fast Approximate Nearest Neighbor Search Algorithm Based on kNN Graph}, 
      author={Cong Fu and Deng Cai},
      year={2016},
      eprint={1609.07228},
      archivePrefix={arXiv},
      primaryClass={cs.CV},
      url={https://arxiv.org/abs/1609.07228}, 
}

@inproceedings{wdong_www11,
author = {Dong, Wei and Moses, Charikar and Li, Kai},
title = {Efficient k-nearest neighbor graph construction for generic similarity measures},
year = {2011},
isbn = {9781450306324},
publisher = {Association for Computing Machinery},
address = {New York, NY, USA},
url = {https://doi.org/10.1145/1963405.1963487},
doi = {10.1145/1963405.1963487},
booktitle = {Proceedings of the 20th International Conference on World Wide Web},
pages = {577–586},
numpages = {10},
location = {Hyderabad, India},
series = {WWW '11}
}

@article{jmuno_pr19,
author = {Vargas Mu\~{n}oz, Javier and Gon\c{c}alves, Marcos A. and Dias, Zanoni and da S. Torres, Ricardo},
title = {Hierarchical Clustering-Based Graphs for Large Scale Approximate Nearest Neighbor Search},
year = {2019},
issue_date = {Dec 2019},
publisher = {Elsevier Science Inc.},
address = {USA},
volume = {96},
number = {C},
issn = {0031-3203},
url = {https://doi.org/10.1016/j.patcog.2019.106970},
doi = {10.1016/j.patcog.2019.106970},
journal = {Pattern Recogn.},
month = dec,
numpages = {10}
}

@InProceedings{bharw_cvpr16,
author = {Harwood, Ben and Drummond, Tom},
title = {FANNG: Fast Approximate Nearest Neighbour Graphs},
booktitle = {Proceedings of the IEEE Conference on Computer Vision and Pattern Recognition (CVPR)},
month = {June},
year = {2016}
}

@article{cfu_vldb19,
author = {Fu, Cong and Xiang, Chao and Wang, Changxu and Cai, Deng},
title = {Fast approximate nearest neighbor search with the navigating spreading-out graph},
year = {2019},
issue_date = {January 2019},
publisher = {VLDB Endowment},
volume = {12},
number = {5},
issn = {2150-8097},
url = {https://doi.org/10.14778/3303753.3303754},
doi = {10.14778/3303753.3303754},
journal = {Proc. VLDB Endow.},
month = jan,
pages = {461–474},
numpages = {14}
}

@article{ymalk_is14,
title = {Approximate nearest neighbor algorithm based on navigable small world graphs},
journal = {Information Systems},
volume = {45},
pages = {61-68},
year = {2014},
issn = {0306-4379},
doi = {https://doi.org/10.1016/j.is.2013.10.006},
url = {https://www.sciencedirect.com/science/article/pii/S0306437913001300},
author = {Yury Malkov and Alexander Ponomarenko and Andrey Logvinov and Vladimir Krylov},
}

@INPROCEEDINGS{deep,
  author={Yandex, Artem Babenko and Lempitsky, Victor},
  booktitle={2016 IEEE Conference on Computer Vision and Pattern Recognition (CVPR)}, 
  title={Efficient Indexing of Billion-Scale Datasets of Deep Descriptors}, 
  year={2016},
  volume={},
  number={},
  pages={2055-2063},
  doi={10.1109/CVPR.2016.226}}

@ARTICLE{sift,
  author={Jégou, Herve and Douze, Matthijs and Schmid, Cordelia},
  journal={IEEE Transactions on Pattern Analysis and Machine Intelligence}, 
  title={Product Quantization for Nearest Neighbor Search}, 
  year={2011},
  volume={33},
  number={1},
  pages={117-128},
  doi={10.1109/TPAMI.2010.57}}

@article{digra,
author = {Jiang, Mengxu and Yang, Zhi and Zhang, Fangyuan and Hou, Guanhao and Shi, Jieming and Zhou, Wenchao and Li, Feifei and Wang, Sibo},
title = {DIGRA: A Dynamic Graph Indexing for Approximate Nearest Neighbor Search with Range Filter},
year = {2025},
issue_date = {June 2025},
publisher = {Association for Computing Machinery},
address = {New York, NY, USA},
volume = {3},
number = {3},
url = {https://doi.org/10.1145/3725399},
doi = {10.1145/3725399},
journal = {Proc. ACM Manag. Data},
month = jun,
articleno = {148},
numpages = {26}
}

@INPROCEEDINGS {cagra,
author = { Ootomo, Hiroyuki and Naruse, Akira and Nolet, Corey and Wang, Ray and Feher, Tamas and Wang, Yong },
booktitle = { 2024 IEEE 40th International Conference on Data Engineering (ICDE) },
title = {{ CAGRA: Highly Parallel Graph Construction and Approximate Nearest Neighbor Search for GPUs }},
year = {2024},
volume = {},
ISSN = {},
pages = {4236-4247},
doi = {10.1109/ICDE60146.2024.00323},
url = {https://doi.ieeecomputersociety.org/10.1109/ICDE60146.2024.00323},
publisher = {IEEE Computer Society},
address = {Los Alamitos, CA, USA},
month =May}

@article{hnsw,
author = {Malkov, Yu A. and Yashunin, D. A.},
title = {Efficient and Robust Approximate Nearest Neighbor Search Using Hierarchical Navigable Small World Graphs},
year = {2020},
issue_date = {April 2020},
publisher = {IEEE Computer Society},
address = {USA},
volume = {42},
number = {4},
issn = {0162-8828},
url = {https://doi.org/10.1109/TPAMI.2018.2889473},
doi = {10.1109/TPAMI.2018.2889473},
journal = {IEEE Trans. Pattern Anal. Mach. Intell.},
month = apr,
pages = {824–836},
numpages = {13}
}

@inproceedings{diskann,
 author = {Jayaram Subramanya, Suhas and Devvrit, Fnu and Simhadri, Harsha Vardhan and Krishnawamy, Ravishankar and Kadekodi, Rohan},
 booktitle = {Advances in Neural Information Processing Systems},
 editor = {H. Wallach and H. Larochelle and A. Beygelzimer and F. d\textquotesingle Alch\'{e}-Buc and E. Fox and R. Garnett},
 pages = {},
 publisher = {Curran Associates, Inc.},
 title = {DiskANN: Fast Accurate Billion-point Nearest Neighbor Search on a Single Node},
 url = {https://proceedings.neurips.cc/paper_files/paper/2019/file/09853c7fb1d3f8ee67a61b6bf4a7f8e6-Paper.pdf},
 volume = {32},
 year = {2019}
}

@article{steiner_hardness,
  author       = {Zeyu Wang and Qitong Wang and Xiaoxing Cheng and Peng Wang and Themis Palpanas and Wei Wang},
  title        = {Steiner-Hardness: A Query Hardness Measure for Graph-Based {ANN} Indexes},
  journal      = {Proc. {VLDB} Endow.},
  volume       = {17},
  number       = {13},
  pages        = {4668--4682},
  year         = {2024},
  doi          = {10.14778/3704965.3704974}
}

@inproceedings{amsaleg_lid,
  author       = {Laurent Amsaleg and Oussama Chelly and Teddy Furon and St\'{e}phane Girard and Michael E. Houle and Ken-ichi Kawarabayashi and Michael Nett},
  title        = {Estimating Local Intrinsic Dimensionality},
  booktitle    = {Proceedings of the 21th {ACM} {SIGKDD} International Conference on Knowledge Discovery and Data Mining ({KDD})},
  pages        = {29--38},
  publisher    = {ACM},
  year         = {2015}
}

@inproceedings{houle_lid,
  author       = {Michael E. Houle},
  title        = {Local Intrinsic Dimensionality {I}: An Extreme-Value-Theoretic Foundation for Similarity Applications},
  booktitle    = {Similarity Search and Applications ({SISAP})},
  series       = {Lecture Notes in Computer Science},
  volume       = {10609},
  pages        = {64--79},
  publisher    = {Springer},
  year         = {2017}
}

@inproceedings{aumuller_lid_bench,
  author       = {Martin Aum\"{u}ller and Matteo Ceccarello},
  title        = {The Role of Local Intrinsic Dimensionality in Benchmarking Nearest Neighbor Search},
  booktitle    = {Similarity Search and Applications ({SISAP})},
  series       = {Lecture Notes in Computer Science},
  volume       = {11807},
  pages        = {113--127},
  publisher    = {Springer},
  year         = {2019}
}

@article{radovanovic_hubness,
  author       = {Milos Radovanovi\'{c} and Alexandros Nanopoulos and Mirjana Ivanovi\'{c}},
  title        = {Hubs in Space: Popular Nearest Neighbors in High-Dimensional Data},
  journal      = {Journal of Machine Learning Research},
  volume       = {11},
  pages        = {2487--2531},
  year         = {2010}
}

@inproceedings{li_adaptive_et,
  author       = {Conglong Li and Minjia Zhang and David G. Andersen and Yuxiong He},
  title        = {Improving Approximate Nearest Neighbor Search through Learned Adaptive Early Termination},
  booktitle    = {Proceedings of the 2020 {ACM} {SIGMOD} International Conference on Management of Data},
  pages        = {2539--2554},
  publisher    = {ACM},
  year         = {2020}
}

@misc{tao,
  author       = {Kaixiang Yang and Hongya Wang and Bo Xu and Wei Wang and Yingyuan Xiao and Ming Du and Junfeng Zhou},
  title        = {Tao: A Learning Framework for Adaptive Nearest Neighbor Search using Static Features Only},
  year         = {2021},
  eprint       = {2110.00696},
  archivePrefix = {arXiv},
  primaryClass = {cs.DB},
  howpublished = {arXiv:2110.00696}
}

@misc{adaptive_beam,
  author       = {Yousef Al-Jazzazi and Haya Diwan and Jinrui Gou and Cameron Musco and Christopher Musco and Torsten Suel},
  title        = {Distance Adaptive Beam Search for Provably Accurate Graph-Based Nearest Neighbor Search},
  year         = {2025},
  eprint       = {2505.15636},
  archivePrefix = {arXiv},
  primaryClass = {cs.DS},
  howpublished = {arXiv:2505.15636}
}

@inproceedings{spacev,
  author       = {Harsha Vardhan Simhadri and George Williams and Martin Aum{\"{u}}ller and Matthijs Douze and Artem Babenko and Dmitry Baranchuk and Qi Chen and Lucas Hosseini and Ravishankar Krishnaswamy and Gopal Srinivasa and Suhas Jayaram Subramanya and Jingdong Wang},
  title        = {Results of the {NeurIPS'21} Challenge on Billion-Scale Approximate Nearest Neighbor Search},
  booktitle    = {Proceedings of the {NeurIPS} 2021 Competitions and Demonstrations Track},
  series       = {Proceedings of Machine Learning Research},
  volume       = {176},
  pages        = {177--189},
  publisher    = {PMLR},
  year         = {2021}
}

@inproceedings{imagenet,
  author       = {Jia Deng and Wei Dong and Richard Socher and Li-Jia Li and Kai Li and Li Fei-Fei},
  title        = {{ImageNet}: A Large-Scale Hierarchical Image Database},
  booktitle    = {2009 {IEEE} Conference on Computer Vision and Pattern Recognition ({CVPR})},
  pages        = {248--255},
  year         = {2009},
  doi          = {10.1109/CVPR.2009.5206848}
}
